\documentclass[letterpaper]{article}
\usepackage[left=2.4cm, top=2.4cm, bottom=2.4cm, right=2.4cm]{geometry}
\usepackage[utf8]{inputenc}
\usepackage[american]{babel}
\usepackage{amsmath,amssymb,amsthm,latexsym,float,epsfig,bbm,amsrefs,mathtools}
\usepackage{subcaption}
\usepackage{todonotes}
\usepackage{enumitem}
\usepackage{graphicx}
\usepackage{svg}
\usepackage{color,soul}
\usepackage{rotating}
\usepackage{xspace}
\numberwithin{equation}{section}

\def\R{\mathbb{R}}

\def\C{\mathbb{C}}
\def\E{\mathbb{E}}
\def\P{\mathbb{P}}
\def\1{\mathbbm{1}}

\def\F{\mathcal{F}}

\def\Lsp{\mathcal{L}}
\def\calL{\mathcal{L}}
\def\Mt{\widetilde{M}}
\def\Gt{\widetilde{G}}
\def\posc{\eta}

\newcommand{\Expec}[1]{\E\left[\, #1 \, \right]}
\newcommand{\Expecnb}[1]{\E #1 }

\newcommand{\proc}[1]{\left( #1 \right)_{t \geq 0}}

\newcommand{\Ps}{\P^{*}}

\newcommand{\Probs}[1]{\P^{*}\left( #1 \right)}
\newcommand{\LL}[2]{\calL_{#1}^{#2}}

\newcommand{\levy}{L\'{e}vy\xspace}
\newcommand{\fl}{Figueroa-L{\'o}pez }
\newcommand{\flns}{Figueroa-L{\'o}pez}
\newcommand{\tht}[1]{\theta_{0}\left( #1 \right)}
\newcommand{\ltrip}{\left(b,0,\nu \right)}

\def\abs#1{\left\lvert #1 \right\rvert}
\def\rv[#1]#2{RV_{#1}^{#2}}

\theoremstyle{plain}
\newtheorem{thm}{Theorem}[section]
\newtheorem{lem}[thm]{Lemma}
\newtheorem{prop}[thm]{Proposition}
\newtheorem{cor}[thm]{Corollary}
\newtheorem{fclaim}[thm]{Formal Claim}
\theoremstyle{definition}

\theoremstyle{remark}
\newtheorem{rem}[thm]{Remark}

\usepackage{tcolorbox}
\tcbuselibrary{breakable}
\usepackage[colorlinks]{hyperref}
\newcommand{\RE}[1]{\Re{\left( #1 \right)}}

\title{Higher-order ATM asymptotics for the CGMY model via the characteristic function}
\author{Allen Hoffmeyer\thanks{Penumbra Investment Group, New York, NY 10020, USA ({\tt allen.hoffmeyer@gmail.com}).}
\and Christian Houdr\'e\thanks{School of Mathematics, Georgia Institute of 
    Technology, Atlanta, GA 30332, USA ({\tt houdre@math.gatech.edu}).}
   \thanks{Research supported in part by grants \#524678 and MP-TSM-00002660 from the Simons Foundation.}} 

\date{\today}

\begin{document}

\maketitle

\begin{abstract}
Using only the characteristic function, we derive short-time at-the-money
(ATM) call-price asymptotics for the exponential CGMY model with activity
parameter $Y\in(1,2)$.  The Lipton--Lewis formula expresses the normalized
ATM call price, denoted $c(t,0)$, in terms of the characteristic exponent,
which, upon rescaling at the rate $t^{-1/Y}$ from the $Y$-stable domain of
attraction, yields $c(t,0) = d_{1} t^{1/Y} + d_{2} t + o(t)$ as
$t\downarrow 0$.  The first-order coefficient $d_{1}$ is the known stable
limit from the domain of attraction of a symmetric $Y$-stable law, and
$d_{2}$ is given by an explicit integral involving the characteristic
exponent and the limiting stable exponent.  We then extract closed-form
higher-order coefficients by keeping the full Lipton--Lewis integrand intact
and introducing a dynamic cutoff that partitions the domain into inner,
core, and tail regions, establishing the expansion with controlled
remainder.  All coefficients are verified numerically against existing
closed-form expressions where available.
\end{abstract}

\vspace{0.2 cm}
\noindent{\textbf{AMS 2020 subject classifications}: 60E10, 60F99, 60G51, 91G20, 91G60.}

\vspace{0.1 cm}
\noindent{\textbf{Keywords and phrases}: CGMY models; short-time asymptotics;
ATM option pricing; characteristic function; Lipton--Lewis formula;
implied volatility; stable domain of attraction.}

\section{Introduction}
\label{sec:intro}

The CGMY model, introduced by Carr, Geman, Madan, and Yor \cite{cgmy},
provides a flexible family of pure-jump \levy processes indexed by four
parameters: a jump intensity $C>0$, exponential damping rates $G,M \geq 0$
governing the left and right tails, and an activity parameter $Y\in(0,2)$
controlling the fine structure of the jump paths.  When $1<Y<2$, the process
has infinite activity and infinite variation, and the \levy measure near the
origin resembles that of a symmetric $Y$-stable law up to exponential
tempering.  This makes CGMY a natural testing ground for short-time
option-price asymptotics, where the behavior of small jumps drives the
leading-order terms and the tempering enters only at subleading orders.

The short-time behavior of at-the-money (ATM) option prices in exponential
\levy models has received considerable attention over the last two decades.
For pure-jump models whose \levy measure has stable-like small jumps, the
first-order ATM call price scales as $t^{1/\alpha}$, where $\alpha$ is the
Blumenthal--Getoor index of the process.  Early results in this direction
were obtained by Tankov \cite{tankov_review}, who gave first-order price and
implied volatility asymptotics for general exponential \levy models, and by
\fl and Forde \cite{figueroa_forde}, who determined the leading-order ATM
coefficient for CGMY specifically via the small-maturity smile.  Muhle-Karbe
and Nutz \cite{karbe_nutz} established first-order ATM asymptotics for
general martingales with stable-like small jumps, connecting ATM prices to
first absolute moments and covering CGMY as a special case.  When a Brownian
component is present, the jump contribution is always lower order and the ATM
price scales as $\sqrt{t}$, driven entirely by the Gaussian part of the
triplet.

The second- and third-order picture is substantially richer.  \flns, Gong,
and Houdr\'{e} \cite{fl_houdre_cgmy} derived the expansion
\begin{align}
    c(t, 0) = d_{1} t^{1/Y} + d_{2} t + 
    \begin{cases} 
    a_{2,1}t^{2-\frac{1}{Y}} + o\left( t^{2 - \frac{1}{Y}} \right), & 
    \text{if } 1 < Y \leq \frac{3}{2} \\
    a_{1,2} t^{2/Y} + o\left( t^{2/Y} \right), & \text{if } \frac{3}{2}
    \leq Y < 2,
    \end{cases}
    \label{eq:fl_third_order}
\end{align}
for the pure-jump CGMY model, using a change of probability measure under
which the CGMY process becomes stable, together with high-order density
expansions for stable laws.  The bifurcation at $Y=3/2$ in the third-order term was a new
phenomenon: while the first two orders of the expansion ($t^{1/Y}$
and $t$) involve no piecewise behavior in $Y$, the third-order
exponent splits into $2-1/Y$ for $Y \leq 3/2$ and $2/Y$ for
$Y \geq 3/2$. In a companion paper, \flns, Gong, and Houdr\'{e}
\cite{figueroa_lopez_gong_houdre_2014} extended the second-order expansion to
a broader class of tempered \levy processes whose \levy measure takes the form
$\abs{x}^{-Y-1} q(x)$ for suitable decay functions $q$.  The
close-to-the-money regime (where the strike converges to the spot price as
maturity tends to zero) was subsequently treated by \fl and
\'{O}lafsson \cite{figueroa_lopez_olafsson_2016}, who relaxed the regularity
conditions on the \levy density to the weakest possible and replaced the
Brownian component with an independent stochastic volatility process with
leverage.  The same authors \cite{figueroa_lopez_olafsson_skew} provided
high-order expansions for the ATM implied volatility skew under stochastic
volatility with stable-like jumps, showing how leverage and vol-of-vol
parameters enter the skew at different orders.  Related ATM implied volatility
slope asymptotics for \levy models with a Brownian component were obtained by
Gerhold, G\"{u}l\"{u}m, and Pinter \cite{gerhold_2016}, who showed that for
CGMY with $Y\in(1,2)$ the ATM slope blows up at a rate depending on $Y$.

From a different perspective, Andersen and Lipton \cite{andersen_lipton}
surveyed the asymptotic theory for exponential \levy processes in a
comprehensive treatment that covers short-time, long-time, and wing
asymptotics for tempered stable models.  Their work emphasizes the role of
the Lipton--Lewis (LL) formula, which expresses option prices as Fourier
integrals involving the characteristic exponent, and they sketch how this
representation can be used to obtain first-order ATM asymptotics for CGMY.
However, as we discuss in Section~\ref{sec:cgmy_second}, their argument
requires further justification at one point (the claim that a certain integral
is $O(\varepsilon)$ after letting $t \rightarrow 0$), and a rigorous
second-order derivation from the characteristic function was not carried out.

Since a \levy process is completely and uniquely described by its
characteristic function, it is natural to ask whether the ATM call-price
asymptotics, including the second-order term, can be obtained directly from
the characteristic exponent via the LL representation, without passing through
measure transformations or density expansions.  The present paper addresses
this question.  We provide a second-order ATM expansion for CGMY
with $Y\in(1,2)$ that works entirely within the Fourier framework.  The
Lipton--Lewis formula expresses the normalized ATM call price in terms of the
characteristic exponent, and by rescaling at the stable concentration scale
$u \sim t^{-1/Y}$ we show that $c(t,0) = d_{1} t^{1/Y} + d_{2} t + o(t)$ as
$t\downarrow 0$, where $d_{1}$ is the known first-order coefficient and
$d_{2}$ is given by an explicit integral involving the characteristic exponent
and the limiting stable exponent $\theta_{0}$.  We verify numerically that
$d_{2}$ coincides with the closed-form second-order coefficient from
\cite{figueroa_lopez_gong_houdre_2014}.

We then extend the combined-integrand method to extract closed-form
coefficients beyond second order.  In the outer (large-$w$) region, the
integrands reduce to classical Laplace-type integrals against the factors
$e^{-\sigma_Y tw^Y}$ and $(1 - e^{-\sigma_Y tw^Y})$, which evaluate in
closed form via gamma functions.  The drift-squared coefficient $a_{2,1}$
at $t^{2-1/Y}$ agrees with the expression
$\tilde{\gamma}^2 p_Z(1,0)/2$ from \cite{fl_houdre_cgmy}, providing an
independent derivation via the Lipton--Lewis representation.  The first
binomial coefficient $a_{1,2}$ at $t^{2/Y}$ is a new closed-form result.
We also identify a quartic drift term $a_{4,1}$ at $t^{4-3/Y}$,
proportional to $\tilde{b}^4$, which enters the expansion for
$Y < 5/4$.  An important structural finding is that odd powers of the
drift $(i\tilde{b}\,w)^{2j+1}$ are purely imaginary and contribute
nothing to the expansion; this eliminates the exponent $3 - 2/Y$ and
shifts the effective bifurcation from $Y = 4/3$ to $Y = 5/4$.

The resulting five-term expansion
\[ c(t,0) = d_1\,t^{1/Y} + d_2\,t + a_{2,1}\,t^{2-1/Y}
+ a_{4,1}\,t^{4-3/Y} + a_{1,2}\,t^{2/Y}
+ o(t^{\max(2-1/Y,\,2/Y)})
\]
is verified numerically to high precision. At second order, 
splitting the two terms in the integrand and 
handling them separately identifies the correct orders but gives
the wrong coefficients, as the cancellation near $v = 0$ survives
only when the full integrand is kept intact.

More recently, the CGMY framework and its short-time asymptotics have found
applications in several directions.  Zhao and Li \cite{zhao_li_2022} and 
Azzone and Baviera \cite{azzone_2024} derived
first-order ATM expansions and implied volatility asymptotics for additive
normal tempered stable processes, a time-inhomogeneous generalization of
CGMY.  \flns, Gong, and Han \cite{figueroa_lopez_gong_han_2022} developed
estimation methods for tempered stable models of infinite variation that
exploit the short-time asymptotic structure for calibration.  Forde, Smith,
and Viitasaari \cite{forde_smith_viitasaari} established small-time Edgeworth
expansions for near-the-money options under rough stochastic volatility plus
generalized tempered stable jumps; when the rough volatility is turned off,
their result reduces to a variant of the third-order CGMY expansion from
\cite{fl_houdre_cgmy}.  \flns, Gong, and Lorig
\cite{figueroa_lopez_gong_lorig} extended short-time call-price expansions to
leveraged ETFs driven by \levy processes with local volatility.

This article is organized as follows.
Section~\ref{sec:prelim} recalls the necessary background on \levy processes,
exponential \levy models, and the CGMY process.
Section~\ref{sec:cgmy_second} presents the Lipton--Lewis formula and derives
both the first- and second-order ATM expansions from the characteristic
function. Section~\ref{sec:formal_expansion} develops the expansion beyond second
order: we first derive closed-form coefficients $a_{2,1}$, $a_{4,1}$,
and $a_{1,2}$ via a formal Laplace method, then prove
the five-term expansion using a dynamic cutoff argument, and verify
the coefficients numerically.  We also analyze the full exponent lattice
including its bifurcation structure and the vanishing of odd drift powers.
Section~\ref{sec:conclusion} collects concluding remarks and directions for
future work.

\section{Preliminaries}
\label{sec:prelim}

We start by briefly recalling the basic material on \levy processes and exponential-\levy 
models used in what follows; see \cites{sato,applebaum,cont_tankov} for comprehensive accounts.

A stochastic process $\proc{X_t}$ on $(\Omega,\F,\P)$ with values in $\R$ is a 
\emph{\levy process} if it has stationary and independent increments, càdlàg paths, 
and $X_0=0$ a.s.  Every \levy process is uniquely characterized 
by its \emph{characteristic triplet} $(b,\sigma,\nu)$, where $b\in\R$, $\sigma\ge0$, and 
$\nu$ is a positive Borel measure on $\left(\R, \mathcal{B}(\R) \right)$ without
atom at the origin satisfying 
$\int_{\scriptscriptstyle \R} (1\wedge x^2)\nu(dx)<\infty$. Hence 
$\Expec{e^{iuX_t}} = e^{t\psi(u)}$ where the characteristic exponent is
\begin{align}
    \psi(u)
      = iub - \tfrac12\sigma^2u^2
        + \int_{\R}
            \big( e^{iux}-1 - iu x \1_{\{\abs{x}\le1\}} \big)\,\nu(dx),
        \qquad u\in\R. \label{eq:lk_formula}
\end{align}
This is compactly written as $X_t\sim \Lsp(b,\sigma,\nu)$.

Under mild conditions, $X_t$ can be written as
\[
  X_t = bt + \sigma W_t
        + \int_{\abs{x}\le1} x\,\tilde N(t,dx)
        + \int_{\abs{x}>1} x\, N(t,dx),
\]
where $W$ is a Brownian motion, $N$ is a Poisson random measure with intensity $\nu(dx)\,dt$, 
and $\tilde N$ is its compensated version. 

If $\Ps$ is an equivalent measure defined by
\[
    \frac{d\Ps}{d\P}\Big|_{\F_t}
      = \exp\!\big(\theta X_t - t\psi(-i\theta)\big),
\]
then $\proc{X_t}$ remains a \levy process under $\Ps$, with transformed triplet
$(b_\theta,\sigma,\nu_\theta)$ given by the Esscher transform
$\nu_\theta(dx)=e^{\theta x}\nu(dx)$ and $b_\theta=b+\sigma^2\theta
+\int_{\abs{x}\le1}x(e^{\theta x}-1)\nu(dx)$.  Such changes of measure appear naturally in 
risk-neutral valuation.

In an exponential \levy model the asset price is
\[
  S_t = S_0 e^{X_t}, \qquad X_t\sim\Lsp(b,\sigma,\nu),
\]
so that $\E[S_t]=S_0 e^{t\psi(-i)}$, where $\psi$ is the \levy--Khintchine
exponent associated with the triplet $(b,\sigma,\nu)$.  Since $\psi$ is a 
priori defined only for real arguments, evaluating it at $-i$ requires 
justification.  For the exponential moment $\Expec{e^{X_t}}$ to be finite for
all $t\geq 0$, it is necessary and sufficient that
\begin{equation}
  \int_{\{\abs{y}>1\}} e^{y}\,\nu(dy) < \infty, \label{eq:levy_model_moment_condition}
\end{equation}
which we assume throughout.  Under this condition, the \levy--Khintchine
integral \eqref{eq:lk_formula} converges when $u$ is replaced by $u-iz$ 
for any $z\in[0,1]$, so $\psi$ extends analytically to the strip 
$\{u\in\C:-1\le\Im(u)\le 0\}$; in particular $-i$ lies in this strip and 
$\psi(-i)$ is well defined (see \cite{sato}, Theorem 25.17).  Under a
risk-neutral measure the discounted asset $e^{-rt}S_t$ must be a martingale,
so $\Expec{e^{X_t}}=e^{rt}$ and hence $\psi(-i)=r$.  In this paper we work in the
zero-interest (or forward) case $r=0$, so that $\psi(-i)=0$ and the drift
parameter is fixed by the martingale condition
\begin{align}
        b = -\frac{\sigma^{2}}{2} - \int_{-\infty}^{\infty} \left( e^{y} 
        -1 - y \1_{\left\{ \abs{y} \leq 1 \right\} } \right) \nu(dy). 
        \label{eq:b_martingale_condition}
    \end{align}
    The above conventions (triplet $(b,\sigma,\nu)$, truncation $x\1_{\abs{x}\le1}$, 
    and characteristic exponent $\psi$) are maintained throughout the paper.

The asymptotics of at-the-money option prices and implied volatility are the 
main objects of study in this manuscript. For this purpose, we discuss a few results
that will be necessary in deriving both first- and second-order asymptotics. 

We begin with the short-time behavior of the at-the-money call price
\[
    c(t,0) = \Expecnb{\left( e^{X_{t}} - 1 \right)_{+}},
\]
interpreted under the risk-neutral measure. 
To this end, we use a slightly more convenient representation of the function 
$c$ due to Carr and Madan (see \cite{carr_madan_fft} and \cite{figueroa_forde}). 
We work with the share measure $\Ps$ obtained from the Esscher
transform with parameter $\theta=1$ (so that $d\Ps/d\P|_{\F_t}=e^{X_t}$ when
$\psi(-i)=0$), which satisfies, for all Borel sets $D \subset \R$,
\[
    \Probs{X_t \in D} = \Expecnb{e^{X_{t}} \1_{\{X_t \in D\}}}.
\]
Then \cite{carr_madan} showed the following.

\begin{thm}
    \label{thm:carr_madan}
    Under $\Ps$, let $E$ be a mean $1$ exponential random variable that is independent of 
    $\proc{X_{t}}$. Then, 
    \begin{align}
        \frac{1}{S_{0}} \Expecnb{\left( S_{t} - K \right)_{+}} =
        \Probs{X_{t} - E > \log{\left( \frac{K}{S_{0}} \right)}}.
        \label{eq:option_formula}
    \end{align}
\end{thm}

\begin{cor}
    Letting $K=S_0$, the normalized, at-the-money European call option price 
    has representation
    \begin{align}
        c(t,0) &= \frac{1}{S_{0}} \Expecnb{\left( S_{t} - S_0 \right)_{+}} 
        = \int_{0}^{\infty} e^{-x} \Probs{X_{t} \geq x} dx.
        \label{eq:cm_rep}
    \end{align}
\end{cor}

More generally, option prices can be expressed directly in terms of the 
characteristic exponent of the \levy process.  For instance, the 
Carr--Madan Fourier transform 
formula (see \cite{carr_madan_fft}) gives, for a vanilla option with strike $K$ 
and time to maturity $T$,
\begin{align}
  \frac{C(K,T)}{S_0}
   = e^{-rT}\frac{1}{2\pi}
     \int_{\R}
       e^{-iu\log(K/S_0)}
       \frac{\phi_T(u-i)}{iu(1+iu)}\,du, \label{eq:carr_madan_fourier}
\end{align}
where $\phi_T(u)=\Expec{e^{iuX_T}} = e^{T\psi(u)}$.
A closely related representation due to Lipton and Lewis, which is better 
suited to ATM asymptotics, will be the main tool in 
Section~\ref{sec:cgmy_second}. 

\subsection{CGMY Processes}
The CGMY process (see \cite{cgmy}) is a real-valued \levy process with triplet
$\ltrip$, where $b \in \R$ is arbitrary, and where $\nu$ is given by
\begin{align}
    \nu(dx) = \left( \frac{C e^{-G \abs{x}}}{ \abs{x}^{1+Y}} \1_{x < 0} +
        \frac{C e^{-M x}}{ x^{1+Y}} \1_{x > 0} \right) dx,
        \label{eq:CGMY_nu}
\end{align}
with $C > 0$, $M, G \geq 0$, and $Y<2$. Intuitively, the CGMY process can be 
thought of as a stable-like process where larger jumps are much less
likely. This intuition comes from the \levy measure \eqref{eq:CGMY_nu}, which 
resembles the \levy measure of a stable random variable, save for the inclusion of 
the exponential damping terms. These exponential damping terms serve to decrease 
the intensity of the jumps when $\abs{x}$ is large. We restrict our attention to 
CGMY processes where $1<Y<2$. 

Before proceeding, we note that the exponential damping on both tails of 
\eqref{eq:CGMY_nu} guarantees
\[
\int_{\abs{x}>1} \abs{x}^{p}\,\nu(dx) < \infty 
\qquad\text{for every } p>0,
\]
whenever $M,G > 0$, so that the CGMY process $X_t$ has finite moments of 
all polynomial orders for any positive choice of $M$ and $G$.

The exponential moment condition \eqref{eq:levy_model_moment_condition}, 
which is the additional requirement for $\proc{e^{X_t}}$ to be a 
well-defined exponential \levy model, is more restrictive. From 
\eqref{eq:CGMY_nu},
\[
\int_{\abs{x}>1} e^{x}\,\nu(dx)
    = C \int_{1}^{\infty} \frac{e^{(1-M)x}}{x^{1+Y}}\,dx
      \;+\; C \int_{-\infty}^{-1} \frac{e^{-(1+G)\abs{x}}}{\abs{x}^{1+Y}} \,dx.
\]
The integral over $(-\infty,-1)$ is always finite for any $G\ge 0$.
For the positive tail, when $M>1$ the exponential factor $e^{(1-M)x}$ 
decays and finiteness is immediate; when $M=1$ the integrand reduces to 
$Cx^{-(1+Y)}$, which is integrable on $[1,\infty)$ since $Y>1$. 
Therefore the exponential moment condition holds whenever $M\ge 1$ and 
$1<Y<2$, and $\proc{e^{X_t}}$ is a well-defined exponential \levy model.

In addition to having finite moments of all orders, 
CGMY processes have a simple closed-form characteristic function, given
below as presented in Proposition 4.2 in \cite{cont_tankov}. 
In what follows, we say $X = \proc{X_{t}}$ is a CGMY process whenever $X$ 
is a \levy process where $X_{1}$ has \levy triplet $\ltrip$, $b \in \R$, 
$\nu$ is given by \eqref{eq:CGMY_nu}, and $1<Y<2$.

\begin{prop}
    \label{prop:CGMY_chfn}
    Let $\proc{X_{t}}$ be a CGMY process. Then, for $u \in \R$, 
    its characteristic exponent is given by
    \begin{align}
        \Psi\left( u \right) = t^{-1} \log{\left( \Expec{e^{iuX_{t}}}\right) } 
        =  iu\tilde{b} + C \Gamma\left( -Y \right) \left( \left( M- iu 
        \right)^{Y} + \left( G + iu \right)^{Y} -M^Y - G^Y \right),
        \label{eq:CGMY_exp}
    \end{align}
    where $\tilde{b} = b + \int_{\abs{x}>1} x \nu(dx) $. If 
    \begin{align} 
    \tilde{b} = 
    -C \Gamma\left( -Y \right) \left( \left( M-1 \right)^{Y} + 
    \left( G+1 \right)^{Y} - M^{Y} - G^{Y} \right), \label{eq:CGMY_martingale}
    \end{align} 
    then $\proc{S_{t}} = \proc{ S_{0} e^{X_{t}}}$ is a martingale with 
    respect to its own filtration.
\end{prop}
From here on, we will assume that $\tilde{b}$ is given by 
\eqref{eq:CGMY_martingale} and $M>1$ so that the exponential CGMY process 
$\proc{S_{t}} = \proc{S_{0} e^{X_{t}}}$ is an exponential \levy asset model 
admitting no arbitrage, i.e., $\proc{S_t}$ is a martingale.

\section{CGMY Expansions via Characteristic Functions}
\label{sec:cgmy_second}
When $1<Y<2$, the CGMY process lies in the domain of attraction of a
symmetric $Y$-stable law, and the first-order ATM option asymptotics follow
from the stable limit.  To obtain the second-order ATM correction, we exploit
the Lipton--Lewis (LL) representation, which expresses the call price
entirely in terms of the characteristic function.  This approach isolates the
contribution of the CGMY parameters $(C,G,M,Y)$ in a form suitable for a
controlled expansion; the arguments developed in the appendix establish that
the resulting second-order coefficient is well defined and finite.

We begin by recording a convergence result that underpins the expansion and
then reviewing the Lipton--Lewis formula. From Proposition \ref{prop:CGMY_chfn}, 
a direct computation of the 
characteristic function of $X_t / t^{1/Y}$ shows that the tempering 
terms vanish as $t \downarrow 0$, yielding the following well-known 
convergence result (see, e.g., \cite{figueroa_forde}).
\begin{prop}
    \label{prop:CGMY_convergence}
    Let $\proc{X_{t}}$ be a CGMY process. Then 
    \[
        \frac{X_{t}}{t^{1/Y}} \Rightarrow Z,
    \]
    as $t \rightarrow 0$, where $Z$ is a $Y$-stable, symmetric random variable 
    with characteristic function, given for any $u \in \R$, by
    \begin{align}
        \phi_{Z}\left( u \right) = \Expec{e^{iuZ}} = 
        \exp{\left( - 2 C \Gamma\left( -Y \right) \abs{\cos{ \left( 
            \frac{Y \pi}{2} \right) } } \abs{u}^{Y} \right) }, 
            \label{eq:cgmy_ch_fn}
    \end{align}
    where $\Gamma$ is Euler's gamma function.
\end{prop}

As the third-order expansion \eqref{eq:fl_third_order} already suggests,
there is no simple pattern linking the successive terms in the ATM 
asymptotic series, and even the second-order coefficient depends on all 
four CGMY parameters.  Nevertheless, the Lipton--Lewis representation and 
a careful characteristic-function analysis yield a second-order 
expansion.

\subsection{Lipton--Lewis Formula and Second-Order Expansions}
We now turn to the Lipton--Lewis (LL) formula. For a detailed discussion 
of the LL formula and its 
many applications, we refer the reader to Andersen and Lipton 
\cite{andersen_lipton}.

Below, we introduce the LL formula and demonstrate how it can be 
used to obtain first-order asymptotics for the CGMY process, agreeing with 
the expansion found in \cite{figueroa_forde}. First, we impose 
conditions on the characteristic function necessary to validate the LL 
result. Specifically, the characteristic function has to be well defined as a 
function of a complex variable in a certain domain in $\C$. (As usual,
for $z \in \C$, $\Re{z}$ is its real part, $\Im{z}$ its imaginary part, 
and $\bar{z}$ its complex conjugate.) 

\begin{prop}
    \label{prop:sato_complex_strip}
    Let $\proc{X_{t}}$ be a pure-jump \levy process with \levy triplet 
    $\ltrip$, $b \in \R$. Additionally, let the \levy measure satisfy
    \[
        \int_{\abs{x}>1 } e^{x} \nu(dx) < \infty,
    \]
    and let $\phi_{t}(z) = \Expec{e^{izX_{t}}}$ exist in the complex strip
    \[
        \mathcal{S} = \left\{ z \in \C : -1 \leq \Im{z} \leq  0 \right\},
    \]
    with $\phi_{t}\left( -i \right) = 1$, then the process $S = \proc{S_{t}} 
    = \proc{S_{0} e^{X_{t}}}$ is an arbitrage-free, well-defined exponential 
    \levy model.
\end{prop}

In what follows, we use $k = \log{(S_0/K)}$ to refer 
to \emph{log-moneyness}, where $K$ is
the strike price of the option and $S_0$ is the spot price. The 
following results can be found in \cite{andersen_lipton} Proposition $5.1$, 
\cite{lewis} Theorem $3.5$ and formula $\left( 3.11 \right)$, and 
\cite{lipton_formula} formula $(3)$. 

\begin{thm}
    \label{thm:ll_thm}
    Let $\Psi$ be the characteristic exponent of the \levy process $\proc{X_{t}}$ 
    driving the exponential \levy asset model $\proc{S_{t}} = \proc{S_{0} 
    e^{X_{t}}}$. Let $\Psi$ exist in a domain of $\C$ containing $\mathcal{S}$.
    Then the normalized call price is given by
    \begin{align}
        c(t, k) := \frac{1}{S_{0}}\,\Expec{(S_{t}-K)_{+}} 
        = 1 - \frac{1}{2\pi} \int_{-\infty}^{\infty} \frac{ e^{ 
            t \Psi\left( u - \frac{i}{2} \right) } }{ u^{2} + \frac{1}{4} }
            e^{k\left( iu - \frac{1}{2} \right) } du.
        \label{eq:ll_form}
    \end{align}
\end{thm}
Setting $k=0$ and using $\int_{-\infty}^{\infty} \left( u^{2} + 1/4 
\right)^{-1} du = 2 \pi$ gives the following ATM call-option pricing formula.  
\begin{cor}
    Under the hypotheses of Theorem \ref{thm:ll_thm}, an ATM call option has
    normalized price 
    \begin{align}
        c(t,0) &= \frac{1}{2\pi} \Re{ \left( \int_{-\infty}^{\infty} 
            \frac{1 - e^{ t \Psi\left( u - \frac{i}{2} \right) } 
            }{u^{2} + \frac{1}{4}} du \right) } 
        = \frac{1}{\pi} \Re{\left( \int_{0}^{\infty} \frac{ 1 - 
            e^{ t \Psi\left( u - \frac{i}{2} \right) } }{u^{2} + 
            \frac{1}{4}} du \right) }.  \label{eq:ll_atm} 
    \end{align}
\end{cor}

We now use the LL formula to derive call-price asymptotics for the CGMY 
model.  The proof of Theorem~\ref{thm:CGMY_first_order} follows the 
argument of \cite{andersen_lipton}, but completes it at a point where their 
derivation requires further justification (see the discussion following 
equation~\eqref{eq:0_to_eps}).  We first record a representation of 
the ATM call price that introduces the notation used throughout.
\begin{prop}
    \label{prop:L_rep}
    Let $\proc{X_{t}}$ be a CGMY process with $M>1$ and $\proc{S_t} = 
    \proc{S_0 e^{X_t}}$ be a martingale. Then the normalized call-price function
    satisfies
    \begin{align}
        c(t,0) = t^{1/Y} \calL\left( t \right),
        \label{eq:cgmy_L}
    \end{align}
    where 
    \begin{align}
        \calL\left( t \right) = \frac{1}{\pi} \Re{\left( \int_{0}^{\infty} \frac{ 1 - 
                \exp{\left( \theta(t,v) \right)}}{v^2 + \frac{1}{4}t^{2/Y}}
                dv \right) }, \label{eq:ll_atm_trans}
    \end{align}
    with 
    \begin{align}
        \theta(t,v) &= i v \tilde{b} t^{1-1/Y} + \frac{\tilde{b}}{2} t +
        C \Gamma\left( -Y \right)\Biggl[ 
        \left( \left( M - \frac{1}{2} \right)t^{1/Y} - iv \right)^{Y} + 
        \left( \left( G + \frac{1}{2} \right)t^{1/Y} + iv \right)^{Y} \notag \\
        & \hspace{4.2in} - M^{Y} t - G^{Y} t \Biggr] \notag  \\
        &= i v \tilde{b} t^{1-1/Y} + \kappa t +
            C \Gamma\left( -Y \right)\left( 
            \left( \Mt t^{1/Y} - iv \right)^{Y} + 
            \left( \Gt t^{1/Y} + iv \right)^{Y} \right), \label{eq:exp_fn}
    \end{align}
    $\tilde{b}$ given by \eqref{eq:CGMY_martingale}, and 
    $\kappa = \tilde{b}/2 - C \Gamma(-Y) (M^{Y} + G^{Y})$. Moreover, 
    the real part of \eqref{eq:exp_fn} has the representation
    \begin{align}
        r(t,v) &:= \Re{\left( \theta(t,v) \right) } \notag \\
        &= \kappa t + C \Gamma(-Y)   \left[ \left( \Mt^2   
        t^{2/Y} + v^2 \right)^{Y/2}   
        \cos \left( Y \arctan \left(
        -\frac{v}{\Mt t^{1/Y}} \right) \right) \right. \notag \\
        &\;\;\; + \left. \left( \Gt^2   t^{2/Y} + v^2
        \right)^{Y/2}   \cos \left( Y \arctan \left(
        \frac{v}{\Gt t^{1/Y}} \right) \right) \right].
        \label{eq:real_theta}
    \end{align}
\end{prop}

\begin{thm}
    \label{thm:CGMY_first_order}
    Let $\proc{X_{t}}$ be a CGMY process with $M>1$ and let $\proc{S_{t}} = 
    \proc{S_{0} e^{X_{t}}}$ be a martingale, then the first-order normalized 
    call-price can be represented as
    \begin{align}
        c\left( t, 0 \right) = d_{1} t^{1/Y} + o\left( t^{1/Y} \right),
        \label{eq:cgmy_firstorder}
    \end{align}
    as $t \rightarrow 0$, with $d_{1}$ given by 
    \begin{align}
        d_{1} &= \frac{1}{\pi} \Re{\left( \int_{0}^{\infty} \frac{1 - 
            \exp{\left( \theta_{0}(u) \right)}}{ u^{2}} du \right) } 
            = \frac{1}{\pi} \Gamma\left( 1 - \frac{1}{Y} \right) \left( 
            2 C \Gamma\left( -Y \right) \abs{\cos{\left( \frac{\pi Y}{2} 
            \right) } } \right)^{1/Y}, 
            \label{eq:cgmy_d1} 
    \end{align}
    with 
    \begin{align}
        \theta_{0}\left( u \right) := C \Gamma\left( -Y \right) \left( 
        \left( -i \right)^{Y} + i^{Y} \right) \abs{u}^{Y}.
        \label{eq:cgmy_exp0}
    \end{align}
\end{thm}

We are in a position now to prove second-order call-price asymptotics for 
exponential CGMY processes using only asymptotic expansions involving the 
characteristic function, and we indicate how the same 
methodology can be used to study more general exponential \levy models.

\begin{thm}
    \label{thm:second_order_CGMY}
    Let $\proc{X_{t}}$ be a CGMY process with $M>1$ and let $\proc{S_{t}} = \proc{S_{0} 
    e^{X_{t}}}$ be a martingale, then the second-order normalized call-price can be 
    represented as
    \begin{align}
        c\left( t, 0 \right) = d_{1} t^{1/Y} + d_{2} t + o(t),
        \label{eq:cgmy_secondorder}
    \end{align}
    as $t \rightarrow 0$, where $d_{1}$ is as in Theorem \ref{thm:CGMY_first_order} and 
    \begin{align}
        d_{2} =
        \frac{1}{\pi} \int_{0}^{\infty} \frac{ \left( w^{2} + \frac{1}{4} 
        \right) \tht{w} - w^{2} \Re{\left( \psi_0\left( w \right) \right) }}{ w^{2}\left( 
        w^{2} + \frac{1}{4} \right) } dw, \label{eq:d2_formal}
    \end{align}
    where $\theta_0$ is given by \eqref{eq:cgmy_exp0} and
    $\psi_0(v) := \Psi(v - i/2)$ is the contour-shifted exponent:
    \begin{align}
        \psi_0\left( v \right) = i v \tilde{b} + \kappa +
            C \Gamma\left( -Y \right)\left( \left( \Mt - iv \right)^{Y} + 
            \left( \Gt + iv \right)^{Y} \right), \label{eq:psi_def}
    \end{align}
    with $\Mt = M-\tfrac12$, $\Gt = G+\tfrac12$, and
    $\kappa = \tilde b /2 - C\Gamma(-Y)\big(M^{Y}+G^{Y}\big)$.
\end{thm}

The second-order expansion in Theorem~\ref{thm:second_order_CGMY} 
recovers the well-known small-time behavior of at-the-money call prices 
for the CGMY model.  This provides a streamlined alternative 
derivation that works solely within the Fourier framework: the 
second-order term arises directly from the behavior of the characteristic 
exponent near the origin, without requiring measure transformations or 
density expansions, and the roles of the parameters $C$, $G$, $M$, and 
$Y$ are transparent at each stage.

One natural question is whether the $d_2$ term in Theorem~\ref{thm:second_order_CGMY} 
agrees with the term in the literature. For the case of \cite{fl_houdre_cgmy}, 
the answer is 
affirmative, as we verified numerically by comparing $d_2$ in \eqref{eq:d2_formal} 
with the closed-form expression
\[
d_2^{FL} = \frac{C \Gamma(-Y)}{2} \left( (M-1)^Y - M^Y - (G+1)^Y + G^Y \right),
\]
from \cite{figueroa_lopez_gong_houdre_2014}, which we denote by $d_2^{FL}$ to avoid confusion.

Figure \ref{fig:diffs} shows heatmaps of the difference between $d_2$ and $d_2^{FL}$ for 
fixed values of $Y$ and various ranges of $M$ and $G$. The differences are uniformly 
of order $10^{-7}$, consistent with the 
numerical tolerance of the quadrature, confirming agreement between the two expressions.
\begin{figure}[ht]
    \centering
    \includegraphics[width=0.7\linewidth]{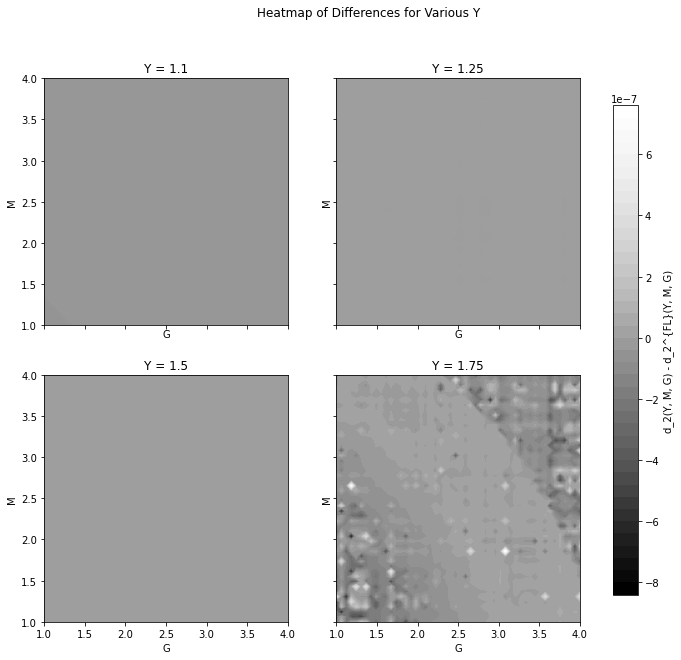}
    \caption{$d_2 - d_2^{FL}$ for various values of $Y$ and ranges of $M$ and $G$ with
    $C=1$.
    All differences are of order $10^{-7}$.}
    \label{fig:diffs}
\end{figure}

\section{Higher-order ATM Expansion via the Characteristic Exponent}
\label{sec:formal_expansion}
In this section, we extract the higher-order terms of the ATM call-price expansion. 
We first present a formal derivation using a combined-integrand approach and Laplace 
approximations to build intuition and identify the closed-form coefficients. We then 
justify this expansion by introducing a dynamic cutoff to 
carefully control the asymptotic error terms.

For a CGMY process with parameters $(C,G,M,Y)$ and $1<Y<2$,
the characteristic exponent is
\[
\Psi(u)
=
i u \tilde b
\;+\;
C\,\Gamma(-Y)
\left[
    (M - iu)^{Y}
  + (G + iu)^{Y}
  - M^{Y}
  - G^{Y}
\right],
\qquad u\in\R,
\]
where
$
\tilde b
=
b
+ \int_{\abs{x}>1} x\,\nu(dx),
$
and under the martingale condition,
\[
\tilde b
=
- C\,\Gamma(-Y)
\left[
    (M - 1)^{Y}
  + (G + 1)^{Y}
  - M^{Y}
  - G^{Y}
\right].
\]

In the Lipton--Lewis representation, we need the exponent evaluated
along the vertical shift $u \mapsto u - i/2$:
\begin{align}
\Psi\!\left(u - \frac{i}{2}\right)
=
i\!\left(u - \frac{i}{2}\right)\tilde b
\;+\;
C\,\Gamma(-Y)
\left[
    \left(\Mt - iu\right)^{Y}
  + \left(\Gt + iu\right)^{Y}
  - M^{Y}
  - G^{Y}
\right], \label{eq:psi}
\end{align}
where, as in Section~\ref{sec:cgmy_second}, $\Mt = M - 1/2$ and 
$\Gt = G + 1/2$.
We expand the terms $(\Mt - iu)^{Y}$ and $(\Gt + iu)^{Y}$ as
\begin{align}
    \left(\Mt - iu\right)^{Y}
      &= (-iu)^{Y} \left( \frac{\Mt}{-iu} + 1 \right)^{Y} \notag \\
      &= (-iu)^{Y} \left( 1 - \frac{Y \Mt}{iu}
        + \frac{Y(Y-1) \Mt^{2}}{2 (iu)^{2}}
        + \cdots \right), \label{eq:M_term} \\
    \left(\Gt + iu\right)^{Y}
      &= (iu)^{Y} \left( \frac{\Gt}{iu} + 1 \right)^{Y} \notag \\
      &= (iu)^{Y} \left( 1 + \frac{Y \Gt}{iu}
        + \frac{Y(Y-1) \Gt^{2}}{2 (iu)^{2}}
        + \cdots \right). \label{eq:G_term}
\end{align}
This gives, as $\abs{u}\to\infty$, the expansion
\begin{align}
    \Psi\!\left(u - \frac i2\right)
    \approx -\sigma_{Y} \abs{u}^{Y}
    + \text{ terms of order } \abs{u}^{Y-1}, \abs{u}^{Y-2}, \ldots
    + \text{ drift term } i \tilde b\,u, \label{eq:psi_approx}
\end{align}
where $\sigma_{Y} = 2C\,\Gamma(-Y)\,\abs{\cos(\pi Y/2)}$ as before.

We start from the Lipton--Lewis ATM formula
\[
    c(t,0)
    = \frac{1}{\pi}\,
      \Re\!\left(
         \int_{0}^{\infty}
         \frac{1 - \exp\!\big(t\,\Psi(u - i/2)\big)}
              {u^{2} + \tfrac14}\,du
      \right),
\]
where $\Psi$ is the characteristic exponent of the \levy process $X$.

To expose the stable domain-of-attraction scaling, we perform the change of
variables $u = v/t^{1/Y}$, $du = t^{-1/Y}\,dv$, and define
\[
    H(t,v)
    := t\,\Psi\!\left( t^{-1/Y}v - \frac{i}{2} \right).
\]
Substituting into the LL representation yields the rescaled form
\[
    c(t,0)
    = \frac{t^{1/Y}}{\pi}\,
      \Re\!\left(
        \int_{0}^{\infty}
           \frac{1 - e^{H(t,v)}}
                {v^{2} + \tfrac14 t^{2/Y}}\,dv
      \right).
\]
Separating the stable part of the exponent,
\[
    H(t,v)
    = -\sigma_{Y} v^{Y} + h(t,v),
    \qquad
    \sigma_{Y} := 2C\,\Gamma(-Y)\,\abs{\cos(\pi Y/2)},
\]
where $h(t,v)$ collects all the correction terms, and setting
$h_{0}(v) := -\sigma_{Y} v^{Y}$, the binomial expansion and the 
drift term give
\begin{align}
    H(t, v)
      &= - \sigma_{Y} v^{Y}
         + t^{1-1/Y} \alpha(v)
         + t^{1/Y} \beta(v)
         + \text{higher-order terms}
      := h_0(v) + h(t,v), \label{eq:asym_H}
\end{align}
where $\alpha(v) = iv\tilde{b}$ is the drift correction and 
$\beta(v) = C\Gamma(-Y)\,Y\,v^{Y-1}\bigl[(\Mt + \Gt)
\cos\tfrac{(Y-1)\pi}{2} + i(\Gt - \Mt)\sin\tfrac{(Y-1)\pi}{2}\bigr]$ 
is the first binomial correction (the $n=1$ term in 
\eqref{eq:M_term}--\eqref{eq:G_term}).

For the denominator, we use the expansion
\begin{align}
    \frac{1}{v^{2} + \tfrac14 t^{2/Y}}
    = \frac{1}{v^{2}}
      \left(
         1 - \frac{t^{2/Y}}{4v^{2}} + O(t^{4/Y})
      \right). \label{eq:denom_asym}
\end{align}
Continuing formally, we expand the exponential as
\begin{align}
    \exp{\left( H(t,v) \right)}
      &= e^{h_0(v)} \exp{\left( h(t,v) \right)} \notag \\
      &= e^{h_0(v)} \left( 1 + h(t,v) + \frac{1}{2} h(t,v)^2 
         + \frac{1}{6} h(t,v)^3 + \cdots \right). 
         \label{eq:exp_H_form}
\end{align}
Combining \eqref{eq:denom_asym} and \eqref{eq:exp_H_form}, the 
ATM call price takes the canonical asymptotic form
\[
    c(t,0)
    = \frac{t^{1/Y}}{\pi}\,
      \Re\!\left(
         \int_{0}^{\infty}
           \frac{1 - e^{-\sigma_{Y} v^{Y}}\,e^{h(t,v)}}{v^{2}}
           \left[ 1 - \frac{t^{2/Y}}{4v^{2}} + O(t^{4/Y}) \right]
         dv
      \right).
\]
To identify the low-order exponents, it suffices to keep only 
$\alpha$ and $\beta$ in $h$.  Expanding the exponential:
\begin{align}
    \exp{(h(t,v))}
      &= 1
         + t^{1-1/Y} \alpha(v)
         + t^{1/Y} \beta(v) \notag \\
      &\quad
         + \frac{1}{2} t^{2-2/Y} \alpha(v)^2
         + t \,\alpha(v) \beta(v)
         + \frac{1}{2} t^{2/Y} \beta(v)^2 \notag \\
      &\quad
         + \frac{1}{6} t^{3-3/Y} \alpha(v)^3 + \cdots .
         \label{eq:exp_h_expansion}
\end{align}
Thus
\begin{align*}
    1 - \exp{(H(t,v))}
      &= 1 - e^{h_0(v)}
         - e^{h_0(v)} \Big( t^{1-1/Y} \alpha(v)
                             + t^{1/Y} \beta(v)
                             + \tfrac12 t^{2-2/Y} \alpha(v)^2
                             + t \,\alpha(v) \beta(v) \\
      &\qquad\qquad\qquad\qquad\;\;
                             + \tfrac12 t^{2/Y} \beta(v)^2
                             + \tfrac16 t^{3-3/Y} \alpha(v)^3
                             + \cdots \Big),
\end{align*}
and each resulting term picks up an additional factor of 
$t^{1/Y}$ from the prefactor.  The exponent of each term is therefore
\begin{align}
    \rho = \frac{1}{Y} + (\text{exponent from } e^{h}) 
    + \frac{2j}{Y}, \label{eq:rho_exponent}
\end{align}
where the first summand comes from the prefactor, the second from 
the expansion of $e^h$ (i.e., $1-1/Y$ per drift factor and 
$1/Y$ per binomial factor), and the third from the expansion of 
the denominator $1/(v^2 + \tfrac{1}{4}t^{2/Y})$.  Setting $j = 0$ 
and keeping only the terms from $\alpha$ and $\beta$, the 
candidate exponents beyond $1/Y$ and $1$ are
\[
    2 - \frac{1}{Y},\quad 
    \frac{2}{Y},\quad 
    1 + \frac{1}{Y},\quad 
    3 - \frac{2}{Y}.
\]
However, not all of these carry nonzero coefficients.  As we show 
in Section~\ref{subsec:formal_d3}, the cubic drift 
$(i\tilde{b}\,v)^3$ is purely imaginary, so the exponent 
$3 - 2/Y$ is absent from the expansion.  More generally, odd 
powers of the drift always vanish under $\Re$.  The effective 
third-order behavior is thus governed by the two exponents 
$2 - 1/Y$ and $2/Y$, which cross at $Y = 3/2$, recovering the 
bifurcation of \eqref{eq:fl_third_order}.

We present the resulting exponents graphically in 
Figure~\ref{fig:Y_order}, distinguishing curves with nonzero 
coefficients (solid) from those whose coefficients vanish 
(dashed).
\begin{figure}[t]
    \centering
    \begin{subfigure}[t]{0.48\textwidth}
        \centering
        \includegraphics[width=\linewidth]{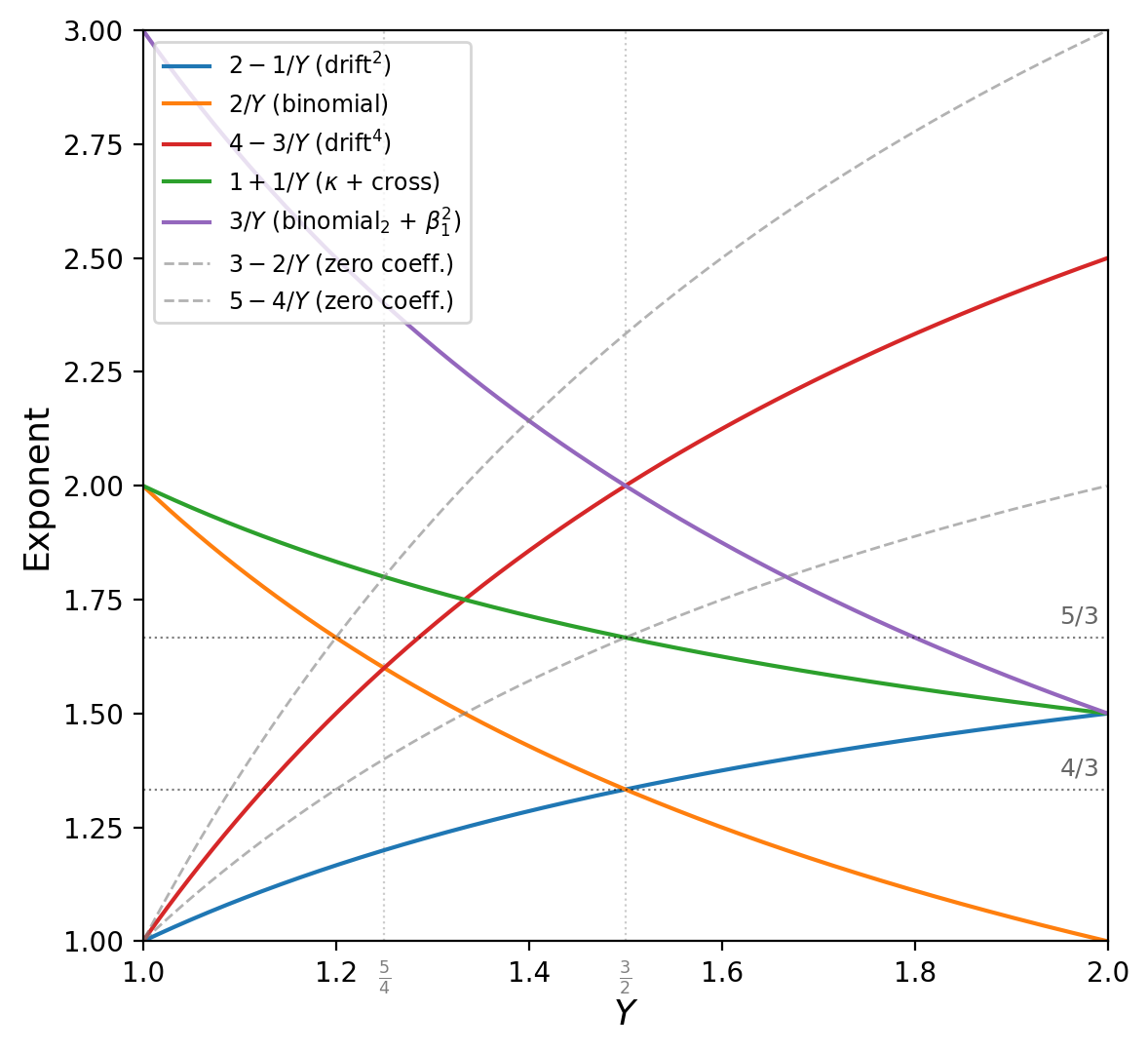}
        \caption{All candidate exponents.  Solid curves carry 
        nonzero coefficients; dashed curves have zero 
        coefficient (odd drift powers).}
        \label{fig:Y_order}
    \end{subfigure}
    \hfill
    \begin{subfigure}[t]{0.48\textwidth}
        \centering
        \includegraphics[width=\linewidth]{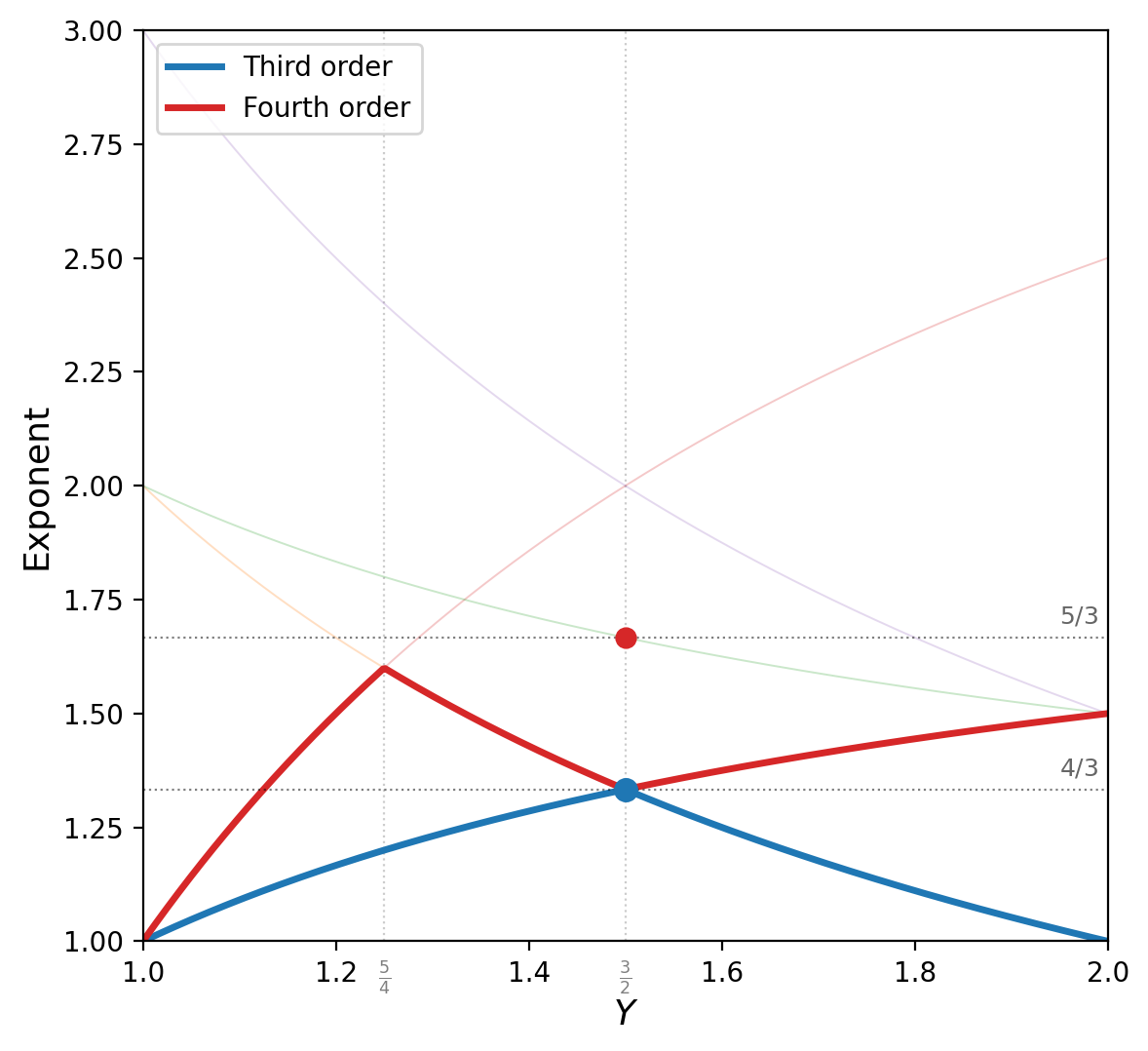}
        \caption{Effective third- and fourth-order exponents.  
        The blue dot at $(3/2,\,4/3)$ marks the coalescence; 
        the red dot at $(3/2,\,5/3)$ marks the fourth-order 
        value at the critical point.}
        \label{fig:Y_order_selected}
    \end{subfigure}
    \caption{Candidate third- and fourth-order exponents as 
    functions of $Y \in (1,2)$.  Vertical dotted lines mark the 
    bifurcation values $Y = 5/4$ and $Y = 3/2$.}
    \label{fig:Y_comp}
\end{figure}

We now turn to deriving the coefficients.  The first-order 
coefficient is
\begin{align}
    d_1 := \frac{1}{\pi} \int_{0}^{\infty} \frac{1 - e^{-\sigma_{Y} u^Y}}{u^2} du. \label{eq:d1}
\end{align}

\begin{prop}[Formal first-order coefficient]
\label{prop:formal_d1}
The first-order coefficient in the formal expansion satisfies
\begin{align}
    d_1 = \frac{1}{\pi} \int_{0}^{\infty} \frac{1 - e^{-\sigma_{Y} u^Y}}{u^2} du 
    = \frac{1}{\pi}\;\Gamma\!\left(1 - \frac{1}{Y}\right) \sigma_{Y}^{1/Y},
    \label{eq:d1_closedform}
\end{align}
where $\sigma_{Y} = 2C\Gamma(-Y)\abs{\cos(Y\pi/2)}$.
\end{prop}

\begin{proof}
Write $I = \int_0^\infty u^{-2}(1 - e^{-\sigma_Y u^Y})\,du$ and 
make the substitution $u = \sigma_Y^{-1/Y}\,s^{1/Y}$, 
$du = \frac{1}{Y}\sigma_Y^{-1/Y}\,s^{1/Y - 1}\,ds$, to obtain
\[
I = \frac{\sigma_Y^{1/Y}}{Y}\int_0^\infty 
\frac{1 - e^{-s}}{s^{1/Y + 1}}\,ds.
\]
By the standard identity $\int_0^\infty s^{\alpha - 1}(1 - e^{-s})\,ds 
= -\Gamma(\alpha)$ for $-1 < \alpha < 0$ (here $\alpha = -1/Y$), 
together with the functional equation 
$-\Gamma(-1/Y) = Y\,\Gamma\!\left(\frac{Y-1}{Y}\right)$, 
we get
\[
I = \frac{\sigma_Y^{1/Y}}{Y}\cdot Y\,
\Gamma\!\left(\frac{Y-1}{Y}\right) 
= \sigma_Y^{1/Y}\,\Gamma\!\left(\frac{Y-1}{Y}\right). \qedhere
\]
\end{proof}

\noindent This agrees with the first-order coefficient from 
Theorem~\ref{thm:CGMY_first_order}, since 
$\sigma_Y^{1/Y} = (2C\Gamma(-Y)\abs{\cos(\pi Y/2)})^{1/Y}$ and 
$\Gamma(1-1/Y) = \Gamma((Y-1)/Y)$.

\subsection{Formal second-order coefficient}

We now determine the formal second-order term $d_2 t$ in the 
expansion $c(t,0) = d_1 t^{1/Y} + d_2 t + \cdots$\,. To that end, 
we work with the rescaled form and set 
$\varepsilon := \tfrac{1}{2}t^{1/Y}$. Write
\begin{align}
    c(t,0) - d_1 t^{1/Y} &= \frac{t^{1/Y}}{\pi}\,\RE{
    \int_0^\infty 
    \left[\frac{1 - e^{H(t,v)}}{v^2 + \varepsilon^2} 
    - \frac{1 - e^{-\sigma_Y v^Y}}{v^2}\right] dv }. 
    \label{eq:formal_remainder}
\end{align}
The key observation is that the two integrands in 
\eqref{eq:formal_remainder} must be kept together.  Combining them 
over a common denominator gives
\begin{align}
    c(t,0) - d_1 t^{1/Y} &= \frac{t^{1/Y}}{\pi}\,\RE{
    \int_0^\infty 
    \frac{v^2\bigl(1 - e^{H(t,v)}\bigr) 
    - (v^2 + \varepsilon^2)\bigl(1 - e^{-\sigma_Y v^Y}\bigr)}
    {v^2(v^2 + \varepsilon^2)}\,dv }.
    \label{eq:formal_combined}
\end{align}
Now substituting $v = t^{1/Y} w$ and using 
$H(t, t^{1/Y} w) = t\,\psi_0(w)$ and 
$-\sigma_Y (t^{1/Y} w)^Y = t\,\tht{w}$, together with the prefactor 
$t^{1/Y}/\pi$ and $dv = t^{1/Y}\,dw$, we obtain
\begin{align}
    c(t,0) - d_1 t^{1/Y} &= \frac{1}{\pi}\,\RE{
    \int_0^\infty 
    \frac{w^2\bigl(1 - e^{t\psi_0(w)}\bigr) 
    - (w^2 + \tfrac{1}{4})\bigl(1 - e^{t\tht{w}}\bigr)}
    {w^2(w^2 + \tfrac{1}{4})}\,dw }.
    \label{eq:formal_rescaled}
\end{align}
For each fixed $w$, expand the exponentials in $t$:
\[
    1 - e^{t\psi_0(w)} = -t\,\psi_0(w) + O(t^2), \qquad 
    1 - e^{t\tht{w}} = -t\,\tht{w} + O(t^2).
\]
At first order in $t$, the numerator in \eqref{eq:formal_rescaled} 
becomes
$-t\bigl[w^2\,\psi_0(w) 
- (w^2 + \tfrac{1}{4})\,\tht{w}\bigr]$,
and taking real parts yields the candidate second-order coefficient
\begin{align}
    d_2^{\text{formal}} &= \frac{1}{\pi}
    \int_0^\infty 
    \frac{(w^2 + \tfrac{1}{4})\,\tht{w} 
    - w^2\,\RE{\psi_0(w)}}
    {w^2(w^2 + \tfrac{1}{4})}\,dw,
    \label{eq:formal_d2}
\end{align}
which agrees with the coefficient $d_2$ from 
Theorem~\ref{thm:second_order_CGMY}.

The integrand in \eqref{eq:formal_d2} is in $L^1[0,\infty)$.  
By partial fractions,
\[
    \frac{(w^2 + \tfrac{1}{4})\,\tht{w} 
    - w^2\,\RE{\psi_0(w)}}
    {w^2(w^2 + \tfrac{1}{4})}
    = \frac{\tht{w}}{w^2} 
    - \frac{\RE{\psi_0(w)}}{w^2 + \tfrac{1}{4}}.
\]
Near $w = 0$, we have $\tht{w}/w^2 = -\sigma_Y w^{Y-2}$, which is 
integrable since $Y - 2 \in (-1,0)$, while 
$\RE{\psi_0(w)}/(w^2 + \tfrac{1}{4})$ is bounded.  Near 
$w = \infty$, the difference 
$\RE{\psi_0(w)} - \tht{w} = O(w^{Y-1})$ gives $O(w^{Y-3})$ decay 
in the combined integrand, which is integrable since 
$Y - 3 < -1$.

One might wish to deal with the expansion by splitting the two 
terms in the numerator of \eqref{eq:formal_remainder} and 
handling them separately.  Such a decomposition identifies the 
correct \emph{orders} (a denominator correction at order $t$ 
and exponent corrections at higher orders) but gives the wrong 
\emph{coefficients}: the second-order term depends only on $C$ 
and $Y$, missing the $G$ and $M$ dependence entirely. The source of the 
shortcoming is that the split is not uniform 
near $v = 0$: each piece is individually 
$O(t^{(Y-1)/Y})$, which strictly dominates $t$ since 
$(Y-1)/Y < 1$, and the cancellation to order $t$ survives only 
when the full integrand is kept intact, as in 
\eqref{eq:formal_combined}. 

\subsection{Formal higher-order coefficients}
\label{subsec:formal_d3}

We now extract the terms beyond $d_2 t$ in the expansion.  By 
\eqref{eq:fl_third_order}, the third-order behavior depends on 
whether $Y$ is above or below $3/2$: the exponent is 
$2 - 1/Y$ for $Y \leq 3/2$ (drift-squared mechanism) and $2/Y$ 
for $Y \geq 3/2$ (first binomial correction).  We show that the 
combined-integrand approach recovers both coefficients $a_{2,1}$ 
and $a_{1,2}$ in closed form, and also identifies an intermediate 
term $a_{4,1}\,t^{4-3/Y}$ from the quartic drift that is relevant 
for $Y < 5/4$.  All three terms are present for all 
$Y \in (1,2)$; only their relative ordering changes at the 
bifurcation points $Y = 3/2$ and $Y = 5/4$.

We denote the higher-order coefficients by $a_{n,m}$, reserving
$d_1$ and $d_2$ for the leading-order coefficients whose notation
is standard in the literature.  The second subscript indicates the
mechanism: $m = 1$ for \emph{drift} terms arising from even powers
of $(i\tilde{b}\,w)^{2k}$, which produce Laplace integrals
against $e^{-\sigma_Y tw^Y}$; and $m = 2$ for \emph{binomial}
terms arising from the tempering corrections
$\beta_n\,w^{Y-n}$, which produce integrals against
$1 - e^{-\sigma_Y tw^Y}$.  The first subscript indicates the
source: $n = 2k$ for the $\tilde{b}^{2k}$ drift power (so
$a_{2,1}$ is the drift-squared coefficient,
$a_{4,1}$ the quartic drift), and $n$ for the $n$-th binomial
correction (so $a_{1,2}$ is the first binomial coefficient).

Define $f(x) := 1 - e^x + x$, so that $f(x) = -x^2/2 + O(x^3)$ 
as $x \to 0$.  Then \eqref{eq:formal_rescaled} can be rewritten 
as
\begin{align}
    R_3(t) := c(t,0) - d_1 t^{1/Y} - d_2 t 
    &= \frac{1}{\pi}\,\RE{
    \int_0^\infty 
    \frac{w^2\,f\!\left(t\psi_0(w)\right) 
    - (w^2 + \tfrac{1}{4})\,f\!\left(t\tht{w}\right)}
    {w^2(w^2 + \tfrac{1}{4})}\,dw },
    \label{eq:R3_integral}
\end{align}
since the linear-in-$t$ terms are exactly those subtracted by 
$d_2 t$.

The global Taylor expansion $f(x) \approx -x^2/2$ applied to the 
integrand in \eqref{eq:R3_integral} gives a candidate at order 
$t^2$, but the resulting integrand has $O(w^{2Y-3})$ decay for 
large $w$, which is not in $L^1[0,\infty)$ when $Y > 1$.  The 
higher-order terms therefore cannot be extracted from a single 
Taylor expansion over all of $[0,\infty)$; instead, they arise 
from the interplay between a Taylor region (small $w$) and a 
Laplace region (large $w$).

Fix a cutoff $\Lambda > 0$ and split 
$R_3(t) = R_3^{\mathrm{in}}(t) + R_3^{\mathrm{out}}(t)$ with 
integrals over $[0,\Lambda]$ and $[\Lambda,\infty)$, respectively.

\medskip\noindent\textbf{Inner region $[0,\Lambda]$.}  
For $w \in [0,\Lambda]$, we have $\abs{t\psi_0(w)} \leq C_\Lambda\,t$ 
uniformly, so the Taylor expansion of $f$ converges uniformly.  
Since the $t^2$ integrand behaves like $w^{2Y-3}$ near infinity 
(and is thus locally integrable on $[0,\Lambda]$ since $2Y-3 > -1$ 
for $Y > 1$), the inner region contributes at most $O(t^2 
\Lambda^{2Y-2})$, which for fixed $\Lambda$ is 
$o(t^{2/Y})$ and $o(t^{2-1/Y})$ since both exponents are 
strictly less than $2$ for $Y \in (1,2)$.

\medskip\noindent\textbf{Outer region $[\Lambda,\infty)$.}  
For $w$ large, $w^2 + \tfrac{1}{4} \approx w^2$, and we write 
$\psi_0(w) = \tht{w} + \delta(w)$ where 
$\delta(w) := \psi_0(w) - \tht{w}$ collects all the tempering 
corrections.  From the binomial expansion 
\eqref{eq:M_term}--\eqref{eq:G_term}, the leading terms in 
$\delta(w)$ for large $w$ are
\[
    \delta(w) = \beta_1\,w^{Y-1} + i\tilde{b}\,w 
    + \kappa + O(w^{Y-2}),
\]
where $\beta_1 = \beta(w)/w^{Y-1}$ is the constant from the 
$n=1$ binomial term, and the drift $i\tilde{b}\,w$ and constant 
$\kappa$ come from the non-binomial parts of $\psi_0$.

The expansion of $f(t\psi_0) - f(t\tht{w})$ is controlled by 
two types of correction: the \emph{linear} correction 
$f'(t\tht{w})\cdot t\,\delta$, which produces terms with the 
factor $(1 - e^{-\sigma_Y tw^Y})$, and the 
\emph{higher-order} corrections 
$f^{(n)}(t\tht{w})\cdot (t\,\delta)^n/n!$ for $n \geq 2$, 
which produce terms with the factor 
$e^{-\sigma_Y tw^Y}$.  We treat the drift-squared 
piece first, since its coefficient can be compared directly 
with existing results.

\medskip\noindent\textbf{Drift-squared piece ($a_{2,1}$).}  
The quadratic correction gives
\begin{align}
    \frac{1}{2}f''(t\tht{w})\cdot(t\,\delta)^2 
    = -\frac{1}{2}\,e^{-\sigma_Y t w^Y}\,(t\,\delta(w))^2.
    \label{eq:f_diff_quad}
\end{align}
The dominant real contribution from $\delta^2$ for large $w$ comes 
from the drift: $(i\tilde{b}\,w)^2 = -\tilde{b}^2\,w^2$.  
(The cross-term $2\beta_1 w^{Y-1}\cdot i\tilde{b}\,w$ 
contributes exponent $1 + 1/Y$, not $2 - 1/Y$, since its 
$w$-power is $Y$; the $\beta_1^2\,w^{2Y-2}$ 
term contributes order $t^{3/Y}$.) 
The leading real contribution from the quadratic correction is
therefore
\begin{align}
    \frac{t^2\,\tilde{b}^2}{2\pi}
    \int_\Lambda^\infty 
    e^{-\sigma_Y t w^Y}\,dw.
    \label{eq:R3_drift}
\end{align}
The standard Laplace integral 
$\int_0^\infty e^{-\lambda u^Y}\,du = 
\lambda^{-1/Y}\,\Gamma(1/Y)/Y$ (with $\lambda = \sigma_Y t$) 
gives
\begin{align}
    \int_0^\infty e^{-\sigma_Y t w^Y}\,dw 
    = \frac{(\sigma_Y t)^{-1/Y}}{Y}\,\Gamma\!\left(\frac{1}{Y}\right),
    \label{eq:laplace_exp}
\end{align}
so the factor $t^2 \cdot (\sigma_Y t)^{-1/Y} = 
\sigma_Y^{-1/Y}\,t^{2-1/Y}$, giving a contribution at order 
$t^{2-1/Y}$ with coefficient
\begin{align}
    a_{2,1} = \frac{\tilde{b}^2\,\sigma_Y^{-1/Y}}{2\pi Y}\,
    \Gamma\!\left(\frac{1}{Y}\right).
    \label{eq:a21_formula}
\end{align}
This agrees with the expression 
$a_{2,1} = \tilde{\gamma}^2 p_Z(1,0)/2$ from 
\cite{fl_houdre_cgmy}, where 
$p_Z(1,0) = \sigma_Y^{-1/Y}\,\Gamma(1/Y)/(\pi Y)$ is the 
density of the symmetric $Y$-stable law $Z_1$ at the origin 
and $\tilde{\gamma} = \tilde{b}$ in the pure-jump case 
($\sigma = 0$).  Indeed, $\tilde{\gamma}^2\,p_Z(1,0)/2 = 
\tilde{b}^2\,\sigma_Y^{-1/Y}\,\Gamma(1/Y)/(2\pi Y) = a_{2,1}$.  
Our derivation via the Laplace integral 
\eqref{eq:laplace_exp} provides an independent route to 
this coefficient directly from the Lipton--Lewis 
representation, without the measure-change and stable-density 
machinery of \cite{fl_houdre_cgmy}.

\medskip\noindent\textbf{Cubic drift ($t^{3-2/Y}$): zero 
coefficient.}
The next candidate from the exponent lattice would be the cubic 
drift at $t^{3-2/Y}$, arising from 
$(i\tilde{b}\,w)^3 = -i\tilde{b}^3\,w^3$, which is purely 
imaginary.  Since $e^{-\sigma_Y tw^Y}$ is 
real, the real part of the integrand vanishes, and the exponent 
$3 - 2/Y$ carries no term in the expansion.  More generally, 
all odd powers $(i\tilde{b}\,w)^{2j+1}$ are purely imaginary,
so the exponents $\{(2j+1) - 2j/Y : j \geq 1\}$ are
absent from the expansion. This is why the line $3 - 2/Y$ 
appears dashed in Figure~\ref{fig:Y_order}.

\medskip\noindent\textbf{Quartic drift piece ($a_{4,1}$).}
The next nonzero even drift power is the quartic.  Proceeding 
exactly as for $a_{2,1}$, the $n = 4$ Taylor coefficient of $f$ 
gives
\[
    \frac{f^{(4)}(t\tht{w})}{4!}\,(t\,\delta)^4 
    = -\frac{1}{24}\,e^{-\sigma_Y t w^Y}\,(t\,\delta(w))^4.
\]
The dominant real contribution from $\delta^4$ for large $w$ is 
$(i\tilde{b}\,w)^4 = \tilde{b}^4\,w^4$, giving the 
integral
\[
    -\frac{\tilde{b}^4\,t^4}{24\pi}
    \int_0^\infty 
    e^{-\sigma_Y t w^Y}\,w^{2}\,dw
    = -\frac{\tilde{b}^4\,t^4}{24\pi}\;
    \frac{(\sigma_Y t)^{-3/Y}}{Y}\,\Gamma\!\left(\frac{3}{Y}\right).
\]
Since $t^4 (\sigma_Y t)^{-3/Y} = \sigma_Y^{-3/Y}\,t^{4-3/Y}$, 
this gives a contribution at order $t^{4-3/Y}$ with coefficient
\begin{align}
    a_{4,1} = -\frac{\tilde{b}^4\,\sigma_Y^{-3/Y}}{24\pi Y}\,
    \Gamma\!\left(\frac{3}{Y}\right).
    \label{eq:a41_formula}
\end{align}
For $Y < 5/4$, the exponent $4 - 3/Y$ lies strictly between 
$2 - 1/Y$ and $2/Y$ (since $4 - 3/Y < 2/Y$ iff $Y < 5/4$), 
so this term is needed in the expansion.  For $Y > 5/4$, 
the exponent exceeds $\max(2-1/Y,\,2/Y)$ and the term is 
absorbed into the remainder.  The coefficient is proportional 
to $\tilde{b}^4$ and is typically very small.

\medskip\noindent\textbf{Binomial piece ($a_{1,2}$).}  
The linear correction gives
\begin{align}
    f(t\psi_0) - f(t\tht{w}) 
    &\approx \bigl(1 - e^{t\tht{w}}\bigr)\,t\,\delta(w).
    \label{eq:f_diff_linear}
\end{align}
The factor $1 - e^{t\tht{w}} = 
1 - e^{-\sigma_Y t w^Y}$ concentrates at large $w$ as 
$t \downarrow 0$.  The dominant 
contribution from $\delta(w)$ for large $w$ is the $n=1$ 
binomial term $\beta_1\,w^{Y-1}$.  The drift term 
$i\tilde{b}\,w$ produces a purely imaginary integrand against 
the real factor $(1 - e^{-\sigma_Y tw^Y})$, so its real part vanishes.  The constant 
$\kappa$ contributes order $t^{1+1/Y}$, which exceeds both 
$t^{2/Y}$ and $t^{2-1/Y}$ for $Y > 1$.  The leading outer 
contribution from the linear correction is therefore
\begin{align}
    \frac{t\,\Re{\beta_1}}{\pi}
    \int_\Lambda^\infty 
    (1 - e^{-\sigma_Y t w^Y})\,w^{Y-3}\,dw.
    \label{eq:R3_binomial}
\end{align}
By the standard Laplace identity, for $-1 < \alpha < 0$,
\begin{align}
    \int_0^\infty (1 - e^{-\lambda u^Y})\,u^{\alpha Y - 1}\,du 
    = -\frac{1}{Y}\,\lambda^{-\alpha}\,\Gamma(\alpha).
    \label{eq:laplace_identity}
\end{align}
With $\alpha = 1 - 2/Y$ (so that $\alpha Y - 1 = Y - 3$ and 
$-\alpha = 2/Y - 1$), the condition $\alpha \in (-1,0)$ holds 
for all $Y \in (1,2)$.  Setting $\lambda = \sigma_Y t$, the 
factor $t \cdot (\sigma_Y t)^{2/Y - 1} = 
\sigma_Y^{(2-Y)/Y}\,t^{2/Y}$, giving a contribution at order 
$t^{2/Y}$ with coefficient
\begin{align}
    a_{1,2} = -\frac{\Re{\beta_1}}{\pi Y}\,
    \sigma_Y^{(2-Y)/Y}\,\Gamma\!\left(1 - \frac{2}{Y}\right)
    = -\frac{C\,\Gamma(-Y)\,(\Mt + \Gt)\,
    \sin\!\frac{Y\pi}{2}}{\pi}\;
    \Gamma\!\left(1 - \frac{2}{Y}\right)\,
    \sigma_Y^{(2-Y)/Y},
    \label{eq:a12_formula}
\end{align}
where we used 
$\Re{\beta_1} = CY\Gamma(-Y)(\Mt + \Gt)\sin(Y\pi/2)$, which 
follows from $\cos\tfrac{(Y-1)\pi}{2} = \sin\tfrac{Y\pi}{2}$.

\medskip
We note that the integrands in \eqref{eq:R3_drift} and 
\eqref{eq:R3_binomial} use the large-$w$ approximation 
$w^2/(w^2 + 1/4) \approx 1$ in place of the exact LL 
denominator.  This is valid because 
the integrands concentrate at 
$w \sim (\sigma_Y t)^{-1/Y} \to \infty$ as $t \downarrow 0$.

\begin{fclaim}[Expansion beyond second order]
\label{prop:formal_d3}
For a CGMY process with $Y \in (1,2)$, the formal 
expansion of the ATM call price is
\begin{align}
    c(t,0) = d_1\,t^{1/Y} + d_2\,t + a_{2,1}\,t^{2-1/Y} 
    + a_{4,1}\,t^{4-3/Y} + a_{1,2}\,t^{2/Y} 
    + o\!\left(t^{\max(2-1/Y,\, 2/Y)}\right),
    \label{eq:formal_five_terms}
\end{align}
as $t \downarrow 0$, where $d_1$ and $d_2$ are as in 
Theorems~\ref{thm:CGMY_first_order} 
and~\ref{thm:second_order_CGMY}, and $a_{2,1}$, $a_{4,1}$, 
$a_{1,2}$ are given by \eqref{eq:a21_formula}, 
\eqref{eq:a41_formula}, and \eqref{eq:a12_formula}, 
respectively.  All three terms are present for all 
$Y \in (1,2)$; for $Y > 5/4$, the quartic drift term 
$a_{4,1}\,t^{4-3/Y}$ is absorbed into the remainder.

The effective ordering depends on $Y$:
\begin{itemize}
\item For $1 < Y < 5/4$: the third order is 
$a_{2,1}\,t^{2-1/Y}$, the fourth is $a_{4,1}\,t^{4-3/Y}$, 
and $a_{1,2}\,t^{2/Y}$ is fifth order.
\item For $5/4 < Y < 3/2$: the third order is 
$a_{2,1}\,t^{2-1/Y}$ and the fourth is $a_{1,2}\,t^{2/Y}$.
\item For $Y = 3/2$: $a_{2,1}$ and $a_{1,2}$ coalesce at 
$t^{4/3}$.
\item For $3/2 < Y < 2$: the third order is 
$a_{1,2}\,t^{2/Y}$ and the fourth is $a_{2,1}\,t^{2-1/Y}$.
\end{itemize}
The cubic drift exponent $3 - 2/Y$ is absent from the 
expansion because $(i\tilde{b}\,w)^3$ is purely imaginary.  
This shifts the effective bifurcation from the lattice 
prediction of $Y = 4/3$ (where $3 - 2/Y$ would cross 
$2/Y$) to $Y = 5/4$ (where the quartic drift $4 - 3/Y$ 
crosses $2/Y$); see Figure~\ref{fig:Y_order_selected}.
\end{fclaim}

\begin{rem}
The coefficient $a_{2,1}$ depends on $\tilde{b}^2$, the squared 
martingale drift, and $a_{4,1}$ is proportional to 
$-\tilde{b}^4$ (the quartic drift correction reduces the 
call price).  Since 
$\tilde{b} = -C\Gamma(-Y)[(M-1)^Y + (G+1)^Y - M^Y - G^Y]$, 
both encode the asymmetry of the tempering; they vanish 
identically when $G = M - 1$, i.e.\ when the shifted parameters 
satisfy $\Gt=\Mt$.

The coefficient $a_{1,2}$ depends on $C$, $Y$, and the sum 
$\Mt + \Gt = M + G$, but not on the difference 
$\Mt - \Gt = M - G - 1$.  The skewness enters $\beta_1$ only 
through its imaginary part, which drops out when taking real 
parts in \eqref{eq:f_diff_linear}.
\end{rem}

The success of the current method at these higher orders relies on 
the assumption that the Taylor expansion of the tempering terms can 
be integrated over the entire domain. However, because the drift term 
$tw$ grows without bound for large $w$, the formal Taylor expansion of 
$e^{t\delta(w)}$ is not uniformly valid on $[0, \infty)$. To elevate 
Formal Claim~\ref{prop:formal_d3} to a theorem, we must introduce 
a dynamic cutoff that separates the region where the Taylor expansion 
holds uniformly from the deep tail where the stable exponential dominates.

\subsection{Proof of the Higher-Order Expansion}
\label{subsec:high_order_proof}

We now prove the higher-order expansion~\eqref{eq:formal_five_terms}, 
corroborating the coefficients derived formally in Section~\ref{subsec:formal_d3}.

\begin{thm}
\label{thm:rigorous_high_order}
Let $\proc{X_t}$ be a CGMY process with $Y \in (1,2)$
satisfying the martingale condition
\eqref{eq:CGMY_martingale}. As $t \downarrow 0$, the
normalized ATM call price admits the asymptotic expansion
\begin{align}
    c(t,0) &= d_1 t^{1/Y} + d_2 t + a_{2,1} t^{2-1/Y} 
    + a_{4,1} t^{4-3/Y} + a_{1,2} t^{2/Y}
    + o(t^{\max\left(2-1/Y, 2/Y \right)}),
    \label{eq:rigorous_expansion}
\end{align}
for $Y \in (7/6,\,2)$ where the coefficients $d_1, d_2, a_{2,1}, a_{4,1},$ and
$a_{1,2}$ are given by \eqref{eq:cgmy_d1},
\eqref{eq:d2_formal}, \eqref{eq:a21_formula},
\eqref{eq:a41_formula}, and \eqref{eq:a12_formula},
respectively. For $Y \in (1,\,7/6]$,
additional even-order drift terms
$a_{2k,1}\,t^{2k-(2k-1)/Y}$ with closed-form
coefficients \eqref{eq:d_drift_general} enter
at or before order $t^{2/Y}$; the
remainder $o(t^{\max(2-1/Y,\,2/Y)})$ holds after
including all such terms through order
$\mathcal{K}(Y) = \left\lfloor 1/(2(Y-1)) \right\rfloor$.
\end{thm}

\begin{rem}[Drift series and uniform expansion]
\label{rem:small_Y_drift}
The even-order pure-drift terms
$a_{2k,1}\,t^{2k-(2k-1)/Y}$ given by
\eqref{eq:d_drift_general} for $k \geq 1$ sum up to
\[
    \Phi_{\mathrm{drift}}(t) :=
    \frac{1}{\pi}\int_0^{\infty}
    \frac{e^{-\sigma_Y tw^Y}\bigl(1 - \cos(\tilde{b}\,tw)\bigr)}
    {w^2}\,dw,
\]
which is a convergent oscillatory Laplace integral
for every $t > 0$ and $Y \in (1,2)$.
Equivalently, the expansion of
Theorem~\ref{thm:rigorous_high_order} can be stated
uniformly for all $Y \in (1,2)$ as
\begin{align}
    c(t,0) = d_1\,t^{1/Y} + d_2\,t
    + \sum_{k=1}^{\mathcal{K}(Y)} a_{2k,1}\,t^{2k-(2k-1)/Y}
    + a_{1,2}\,t^{2/Y}
    + o\!\left(t^{\max(2-1/Y,\, 2/Y)}\right),
    \label{eq:uniform_expansion}
\end{align}
where $\mathcal{K}(Y) = \left\lfloor 1/\bigl(2(Y-1)\bigr)
\right\rfloor \vee 2$ and the coefficients $a_{2k,1}$ are
given by \eqref{eq:d_drift_general}.  The
$k = 1$ and $k = 2$ terms in the sum give the
drift-squared and quartic drift coefficients $a_{2,1}$
and $a_{4,1}$ from the five-term expansion
\eqref{eq:formal_five_terms}.
The $k$-th drift exponent $2k - (2k-1)/Y$ falls at or
before $2/Y$ precisely when
$k \leq 1/\bigl(2(Y-1)\bigr)$; the floor--max ensures
that $a_{2,1}$ and $a_{4,1}$ are always retained (the
latter being harmlessly absorbed into the remainder when
$Y > 5/4$).  For $Y > 7/6$ one has
$\mathcal{K}(Y) = 2$, recovering the five-term expansion
\eqref{eq:formal_five_terms}; as $Y \to 1^{+}$,
$\mathcal{K}(Y) \to \infty$ and the drift terms
proliferate, though they decrease super-factorially in $k$
(by the $(2k)!$ denominator in \eqref{eq:d_drift_general})
and are all proportional to even powers of
$\tilde{b} = -C\Gamma(-Y)[(M-1)^Y + (G+1)^Y - M^Y - G^Y]$.
In particular, they vanish identically when $G = M-1$,
i.e.\ when $\Gt = \Mt$.
\end{rem}

\begin{proof}
Define the exact higher-order remainder
\[
    R_3(t) := c(t,0) - d_1 t^{1/Y} - d_2 t.
\]
By the Lipton--Lewis representation \eqref{eq:ll_atm_trans} and the
stable rescaling $v = t^{1/Y} w$, together with the function
$f(z) := 1 - e^z + z$ (which satisfies $f(z) = -z^2/2 + O(z^3)$
as $z \to 0$ and encodes the subtraction of the first two
asymptotic orders), the remainder admits the exact integral form
\begin{align}
    R_3(t) = \frac{1}{\pi}\,\RE{
    \int_0^\infty
    \frac{w^2\,f(t\psi_0(w)) - (w^2 + \tfrac{1}{4})\,f(t\tht{w})}
    {w^2(w^2 + \tfrac{1}{4})}\,dw}.
    \label{eq:R3_exact}
\end{align}
 
To extract the asymptotic terms, we partition the
integration domain $[0,\infty)$ into three regions.  Fix a large
constant $\Lambda > 0$ and a parameter
$\rho := 1/Y + \varepsilon$ with $\varepsilon > 0$ small enough
that $1 - 1/Y - \varepsilon > 0$ (possible since $Y > 1$).
We write
\[
    R_3(t) = R_3^{\mathrm{in}}(t) + R_3^{\mathrm{core}}(t)
    + R_3^{\mathrm{tail}}(t),
\]
with integrals over $I_{\mathrm{in}} := [0,\Lambda]$,
$I_{\mathrm{core}} := [\Lambda, t^{-\rho}]$, and
$I_{\mathrm{tail}} := [t^{-\rho}, \infty)$, respectively.
 
\medskip
\noindent\textbf{Step 1: The tail region.}
We bound the integrand in \eqref{eq:R3_exact} on
$I_{\mathrm{tail}}$ by working with the combined numerator.
Write the numerator as
\begin{align}
    N(t,w) &:= w^2\,f(t\psi_0(w))
    - \bigl(w^2 + \tfrac{1}{4}\bigr)\,f(t\tht{w}) \notag \\
    &= w^2\bigl[f(t\psi_0(w)) - f(t\tht{w})\bigr]
    - \tfrac{1}{4}\,f(t\tht{w}).
    \label{eq:N_decomp}
\end{align}
For the first piece, since $f(z) = 1 - e^z + z$, we have
\begin{align}
    f(t\psi_0) - f(t\tht{w})
    &= e^{t\tht{w}} - e^{t\psi_0(w)} + t\,\delta(w),
    \label{eq:f_diff_exact}
\end{align}
where $\delta(w) := \psi_0(w) - \tht{w}$.  For $w$ large,
$\RE{\psi_0(w)} \leq -\posc\, w^Y$ for some $\posc > 0$
(by \eqref{eq:real_psi}), so on $I_{\mathrm{tail}}$,
\[
    \abs{e^{t\psi_0(w)}} = e^{t\RE{\psi_0(w)}}
    \leq e^{-\posc\, t w^Y}
    \leq e^{-\posc\, t^{1-\rho Y}}
    = e^{-\posc\, t^{-\varepsilon Y}},
\]
which is super-exponentially small.  The same bound holds for
$e^{t\tht{w}}$.  Therefore $\abs{e^{t\tht{w}} - e^{t\psi_0(w)}}$ is
super-exponentially small on $I_{\mathrm{tail}}$.
 
The $t\,\delta(w)$ term in \eqref{eq:f_diff_exact} is not
exponentially small, but after taking real parts its growth is
controlled.  By Proposition~\ref{prop:A11_est},
$\RE{\delta(w)} = O(w^{Y-1})$ for large $w$
(the drift $i\tilde{b}\,w$ is purely imaginary and drops out
under $\Re$).  Hence
\[
    \RE{w^2\bigl[f(t\psi_0) - f(t\tht{w})\bigr]}
    = O\bigl(e^{-\posc\, t^{-\varepsilon Y}}\bigr)
    + t\,\RE{\delta(w)}\,w^2
    = O\bigl(t\,w^{Y+1}\bigr),
\]
and dividing by $w^2(w^2 + \tfrac{1}{4}) \geq w^4/2$ gives an
integrand bounded by $O(t\,w^{Y-3})$ on $I_{\mathrm{tail}}$.
 
For the second piece in \eqref{eq:N_decomp},
$f(t\tht{w}) = 1 - e^{t\tht{w}} + t\tht{w}$.  Since
$e^{t\tht{w}}$ is super-exponentially small and
$\tht{w} = -\sigma_Y w^Y$, we have
$f(t\tht{w}) = 1 - \sigma_Y t w^Y + O(e^{-\posc\, t^{-\varepsilon Y}})$.
Hence $\tfrac{1}{4}\abs{f(t\tht{w})}/(w^2(w^2+\tfrac{1}{4}))
= O(tw^{Y-4}) + O(w^{-4})$.
 
Combining both pieces, the real part of the integrand on
$I_{\mathrm{tail}}$ is bounded by
$\posc\bigl(t\,w^{Y-3} + w^{-4}\bigr)$.
Integrating over $[t^{-\rho},\infty)$ and using
$Y - 3 \in (-2,-1)$:
\begin{align}
    \abs{R_3^{\mathrm{tail}}(t)}
    &\leq \frac{\posc}{\pi}\left(
    t \int_{t^{-\rho}}^{\infty} w^{Y-3}\,dw
    + \int_{t^{-\rho}}^{\infty} w^{-4}\,dw \right) \notag \\
    &= \frac{\posc}{\pi}\left(
    \frac{t\,(t^{-\rho})^{Y-2}}{2-Y}
    + \frac{(t^{-\rho})^{-3}}{3} \right) \notag \\
    &= O\!\left(t^{1+\rho(2-Y)}\right)
    + O\!\left(t^{3\rho}\right).
    \label{eq:tail_bound}
\end{align}
Since $\rho = 1/Y + \varepsilon$, the first exponent is
$1 + (1/Y + \varepsilon)(2-Y) = 2/Y + \varepsilon(2-Y)$
and the second is $3(1/Y + \varepsilon)$. The first exceeds 
$2/Y$ and the second exceeds $\max(2 - 1/Y,\, 2/Y)$ for all $Y \in (1,2)$,
so $R_3^{\mathrm{tail}}(t) = o(t^{2/Y})$.  Since the
drift-squared integrand involves the factor
$e^{-\sigma_Y t w^Y}$, which is super-exponentially small
on $I_{\mathrm{tail}}$, the tail does not contribute to
the coefficients $a_{2,1}$ or $a_{4,1}$.
 
\medskip
\noindent\textbf{Step 2: The inner region.}
On the fixed bounded interval $[0,\Lambda]$, the functions
$\psi_0(w)$ and $\tht{w}$ are uniformly bounded, so
$t\psi_0(w) \to 0$ and $t\tht{w} \to 0$ uniformly as
$t \downarrow 0$.  We split the integrand by partial fractions:
\[
    \frac{w^2\,f(t\psi_0(w)) - (w^2+\tfrac{1}{4})\,f(t\tht{w})}
    {w^2(w^2+\tfrac{1}{4})}
    = \frac{f(t\psi_0(w))}{w^2+\tfrac{1}{4}}
    - \frac{f(t\tht{w})}{w^2}.
\]
For the first term, $f(z) = -z^2/2 + O(z^3)$ gives
$f(t\psi_0(w)) = O(t^2)$ uniformly on $[0,\Lambda]$, and
$w^2 + \tfrac{1}{4} \geq \tfrac{1}{4}$, so this piece
contributes $O(t^2)$ to the integral.  For the second term,
$\tht{w} = -\sigma_Y w^Y$ gives
$f(t\tht{w}) = -t^2 \sigma_Y^2 w^{2Y}/2 + O(t^3 w^{3Y})$,
whence $f(t\tht{w})/w^2 = O(t^2 w^{2Y-2})$.  Since
$2Y - 2 > 0$ for $Y \in (1,2)$, this is integrable on
$[0,\Lambda]$, and its integral is $O(t^2)$.  Combining:
\[
    R_3^{\mathrm{in}}(t) = O(t^2).
\]
Since $2 > \max(2-1/Y,\, 2/Y)$ for $Y \in (1,2)$, the inner
region contributes only to the remainder.
 
\medskip
\noindent\textbf{Step 3: The core Laplace region.}
On the expanding interval $[\Lambda, t^{-\rho}]$, we write
$\psi_0(w) = \tht{w} + \delta(w)$ where
$\delta(w) = i\tilde{b}\,w + \beta_1\,w^{Y-1} + \kappa
+ O(w^{Y-2})$
collects all the tempering corrections.  We decompose the
numerator \eqref{eq:N_decomp} as
\begin{align}
    N(t,w) &= w^2\bigl[f(t\psi_0) - f(t\tht{w})\bigr]
    - \tfrac{1}{4}\,f(t\tht{w}).
    \label{eq:N_core}
\end{align}
For the first term, using \eqref{eq:f_diff_exact} and the
factorization $e^{t\psi_0} = e^{t\tht{w}}\,e^{t\delta}$, we
obtain the exact identity
\begin{align}
    f(t\psi_0) - f(t\tht{w})
    &= t\,\delta\bigl(1 - e^{t\tht{w}}\bigr)
    - e^{t\tht{w}}\bigl(e^{t\delta} - 1 - t\,\delta\bigr).
    \label{eq:f_identity}
\end{align}
We now expand each factor.  Since
$w \leq t^{-\rho}$ and $\abs{\delta(w)} \leq \posc\,w$ for
$w \geq \Lambda$, the quantity
$\abs{t\,\delta(w)} \leq \posc\,t^{1-\rho}
= \posc\,t^{1-1/Y-\varepsilon} \to 0$
uniformly on $I_{\mathrm{core}}$.  Taylor's theorem with
Lagrange remainder therefore gives, for each fixed $w$,
\begin{align}
    e^{t\delta} - 1 - t\delta
    &= \frac{(t\delta)^2}{2}
    + \frac{(t\delta)^3}{6}
    + \frac{(t\delta)^4}{24}
    + \frac{(t\delta)^5}{120}\,e^{\xi},
    \label{eq:taylor_exp}
\end{align}
where $\abs{\xi} \leq \abs{t\,\delta}$.  Since
$\abs{t\,\delta} \to 0$ uniformly, $e^{\xi}$ is uniformly
bounded, and the fifth-order remainder is
$O(t^5\,w^5)$.
 
We substitute \eqref{eq:taylor_exp} into
\eqref{eq:f_identity}, multiply by $w^2$, divide by
$w^2(w^2 + \tfrac{1}{4})$, and take real parts.  In the
core region, $w \geq \Lambda$ so
$w^2/(w^2 + \tfrac{1}{4}) = 1 + O(w^{-2})$, and we
replace the exact denominator with $w^{-2}$ at the cost
of an error $O(w^{-4})$ in the integrand.  Specifically,
\[
    \frac{w^2}{w^2(w^2+\tfrac{1}{4})}
    = \frac{1}{w^2+\tfrac{1}{4}}
    = \frac{1}{w^2} - \frac{1}{4w^2(w^2+\tfrac{1}{4})},
\]
and the error introduced by this replacement is bounded
by $\posc\,w^{-4}$ times the numerator, which contributes
$O(t^2)$ to the integral (see Step~4 below).
 
We now identify each asymptotic contribution from
the leading terms of \eqref{eq:taylor_exp}, using
$\delta(w) = i\tilde{b}\,w + \beta_1\,w^{Y-1}
+ \kappa + O(w^{Y-2})$.
 
\medskip
\noindent\emph{Linear term
($a_{1,2}\,t^{2/Y}$).}
The first term on the right-hand side of
\eqref{eq:f_identity} gives
\[
    \RE{t\,\delta(w)\,(1 - e^{t\tht{w}})}
    = t\,(1 - e^{-\sigma_Y tw^Y})\,\RE{\delta(w)}.
\]
The drift component $i\tilde{b}\,w$ is purely imaginary
and does not contribute.  The leading real part is
$\RE{\beta_1}\,w^{Y-1}$, producing the integral
\[
    \frac{t\,\RE{\beta_1}}{\pi}
    \int_{\Lambda}^{t^{-\rho}}
    (1 - e^{-\sigma_Y tw^Y})\,w^{Y-3}\,dw.
\]
By the Laplace identity \eqref{eq:laplace_identity}
with $\alpha = 1 - 2/Y \in (-1,0)$ and
$\lambda = \sigma_Y t$, the full-line integral evaluates to
$-(\sigma_Y t)^{2/Y-1}\,\Gamma(1-2/Y)/Y$, giving a
contribution at order $t^{2/Y}$ with coefficient
$a_{1,2}$ as in \eqref{eq:a12_formula}.  The constant
$\kappa$ in $\delta$ produces a Laplace integral at order
$t^{1+1/Y}$, which exceeds both $t^{2/Y}$ and $t^{2-1/Y}$
for $Y > 1$ and hence enters the remainder.
 
\medskip
\noindent\emph{Quadratic term ($a_{2,1}\,t^{2-1/Y}$).}
The quadratic part of \eqref{eq:taylor_exp} gives
\[
    -\RE{e^{t\tht{w}}\,\tfrac{1}{2}(t\delta)^2}
    = -\tfrac{1}{2}\,e^{-\sigma_Y tw^Y}\,
    t^2\,\RE{\delta(w)^2}.
\]
For large $w$, the dominant real contribution from
$\delta^2$ is $(i\tilde{b}\,w)^2 = -\tilde{b}^2\,w^2$,
which gives
\[
    \frac{\tilde{b}^2\,t^2}{2\pi}
    \int_{\Lambda}^{t^{-\rho}}
    e^{-\sigma_Y tw^Y}\,dw.
\]
By \eqref{eq:laplace_exp}, the full-line integral is
$(\sigma_Y t)^{-1/Y}\,\Gamma(1/Y)/Y$, producing a
contribution at order $t^{2-1/Y}$ with coefficient
$a_{2,1}$ as in \eqref{eq:a21_formula}.
 
The cross-term
$2(i\tilde{b}\,w)(\beta_1\,w^{Y-1})$ has real part
$-2\tilde{b}\,\Im(\beta_1)\,w^{Y}$.  Integrated against
$e^{-\sigma_Y tw^Y}$ and the $w^{-2}$ denominator
factor, this produces a contribution at exponent
$1 + 1/Y$, which resides in the remainder.
 
\medskip
\noindent\emph{Cubic term (vanishing).}
The cubic part of \eqref{eq:taylor_exp} contributes
\[
    -\RE{e^{t\tht{w}}\,\tfrac{1}{6}(t\delta)^3}.
\]
The dominant pure-drift component
$(i\tilde{b}\,w)^3 = -i\tilde{b}^3\,w^3$ is purely
imaginary.  Since $e^{-\sigma_Y tw^Y}$ is real, the
real part of this term vanishes identically.  The
remaining cross-terms produce contributions at exponent
$2$ or higher and are absorbed into the remainder.
 
\medskip
\noindent\emph{Quartic term ($a_{4,1}\,t^{4-3/Y}$).}
The quartic part contributes
\[
    -\RE{e^{t\tht{w}}\,\tfrac{1}{24}(t\delta)^4}.
\]
The dominant real contribution is
$(i\tilde{b}\,w)^4 = \tilde{b}^4\,w^4$, giving
\[
    -\frac{\tilde{b}^4\,t^4}{24\pi}
    \int_{\Lambda}^{t^{-\rho}}
    e^{-\sigma_Y tw^Y}\,w^{2}\,dw.
\]
The full-line integral evaluates to
$(\sigma_Y t)^{-3/Y}\,\Gamma(3/Y)/Y$, producing a
contribution at order $t^{4-3/Y}$ with coefficient
$a_{4,1}$ as in \eqref{eq:a41_formula}.
 
\medskip
\noindent\emph{Denominator correction.}
The second term in \eqref{eq:N_core},
$-\tfrac{1}{4}\,f(t\tht{w})$, requires care because
the Taylor bound $\abs{f(t\tht{w})} \leq \posc\,t^{2}w^{2Y}$
holds only where $\sigma_{Y}tw^{Y}$ is bounded, and the
resulting integrand $t^{2}w^{2Y-4}$ fails to be in
$L^{1}[\Lambda,\infty)$ when $Y > 3/2$.  We use instead the
elementary bound $\abs{f(-x)} \leq \min(x^{2}/2,\,x)$ for
$x \geq 0$ and split the core at the Laplace scale
$w_{*} := (\sigma_{Y}t)^{-1/Y}$.  On $[\Lambda,\,w_{*}]$,
where $\sigma_{Y}tw^{Y} \leq 1$, the quadratic bound gives
\[
    \int_{\Lambda}^{w_{*}}
    \frac{\posc\,t^{2}\,w^{2Y}}{w^{2}(w^{2}+\tfrac{1}{4})}\,dw
    \;\leq\;
    \posc\,t^{2}\left(1 + w_{*}^{2Y-3}\right)
    \;=\;
    O\!\left(t^{2}\right) + O\!\left(t^{3/Y}\right).
\]
(The integral $\int_{\Lambda}^{w_{*}} w^{2Y-4}\,dw$ is dominated by $w_{*}^{2Y-3}$
when $Y > 3/2$ and by $\Lambda^{2Y-3}$ when $Y < 3/2$; both contributions
are safely absorbed since $\min(2,\,3/Y) > \max(2/Y,\,2-1/Y)$
for all $Y \in (1,2)$.)
On $[w_{*},\,t^{-\rho}]$, where $\sigma_{Y}tw^{Y} \geq 1$,
the linear bound $\abs{f(-x)} \leq x$ gives
\[
    \int_{w_{*}}^{t^{-\rho}}
    \frac{\posc\,t\,w^{Y}}{w^{2}(w^{2}+\tfrac{1}{4})}\,dw
    \;\leq\;
    \posc\,t\,w_{*}^{Y-3}
    \;=\;
    O\!\left(t^{3/Y}\right).
\]
Since $3/Y > \max(2-1/Y,\,2/Y)$ for all $Y \in (1,2)$
(as $4/Y > 2$ iff $Y < 2$), both contributions are
absorbed into the remainder.
 
\medskip
\noindent\emph{Taylor remainder.}
The fifth-order remainder in \eqref{eq:taylor_exp} is
$R_5 = (t\delta)^5\,e^{\xi}/120$ with $\abs{\xi}\leq\abs{t\delta}$.
Writing $\delta = i\tilde{b}\,w + \delta_{\mathrm{sub}}$
where $\delta_{\mathrm{sub}} := \beta_1\,w^{Y-1} + \kappa
+ O(w^{Y-2})$ collects the non-drift corrections, we
decompose as
\[
    R_5 \;=\; e^{it\tilde{b}w}\bigl(e^{t\delta_{\mathrm{sub}}}
    - P_4(t\delta_{\mathrm{sub}})\bigr)
    \;+\;\bigl(e^{it\tilde{b}w}\,P_4(t\delta_{\mathrm{sub}})
    - P_4(t\delta)\bigr),
\]
where $P_4(z) := \sum_{k=0}^{4} z^k/k!$ denotes the
fourth-order partial exponential.
 
\smallskip
\noindent\emph{Sub-leading remainder.}
Since $\abs{\delta_{\mathrm{sub}}(w)} \leq \posc\,w^{Y-1}$
for $w \geq \Lambda$ and $\abs{e^{it\tilde{b}w}} = 1$, the
first term is such that
\[
\abs{e^{it\tilde{b}w}(e^{t\delta_{\mathrm{sub}}} -
P_4(t\delta_{\mathrm{sub}}))} \leq
\posc\,\abs{t\delta_{\mathrm{sub}}}^5 \leq
\posc\,t^5\,w^{5(Y-1)}.
\]
Integrating against $e^{-\sigma_Y tw^Y}/w^2$,
the integral
\[
\int_{\Lambda}^{\infty} w^{5(Y-1)-2}\,
e^{-\sigma_Y tw^Y}\,dw
=
\begin{cases}
O\!\left(t^{-(5Y-6)/Y}\right), & Y>6/5,\\[2mm]
O(1), & 1<Y\le 6/5,
\end{cases}
\]
produces a contribution of order $t^{5-(5Y-6)/Y} = t^{6/Y}$.
Since $6/Y > 2/Y$ for all $Y > 0$, this term is
absorbed into the remainder for every $Y \in (1,2)$.
 
\smallskip
\noindent\emph{Cross terms.}
The second term,
$e^{it\tilde{b}w}\,P_4(t\delta_{\mathrm{sub}})
- P_4(it\tilde{b}w + t\delta_{\mathrm{sub}})$,
collects monomials
$(it\tilde{b}w)^j\,(t\delta_{\mathrm{sub}})^m/(j!\,m!)$
with $j + m \geq 5$.  We separate the pure-drift
terms ($m = 0$, $j \geq 5$) from the mixed terms
($m \geq 1$, $j \geq 5 - m$).
 
For the mixed terms: each factor of
$t\delta_{\mathrm{sub}}$ contributes $O(tw^{Y-1})$
and each factor of $it\tilde{b}w$ contributes
$O(tw)$, giving an integrand of order
$t^{j+m}\,w^{j + m(Y-1)-2}\,e^{-\sigma_Y tw^Y}$.
By a standard Laplace calculation, the resulting
exponent is $j(Y-1)/Y + (m+1)/Y$.  The minimum over
$j \geq 5 - m$, $m \geq 1$ occurs at $j = 4$,
$m = 1$, giving the order $t^{4(Y-1)/Y + 2/Y} =
t^{4-2/Y}$.  Since $4 - 2/Y > 2/Y$ whenever $Y > 1$,
all mixed terms are absorbed.
 
\smallskip
\noindent\emph{Pure-drift residual.}
The pure-drift terms $\sum_{j \geq 5}
(it\tilde{b}w)^j/j!$ sum to
$e^{it\tilde{b}w} - P_4(it\tilde{b}w)$, whose real
part is
\[
    \cos(\tilde{b}\,tw) - 1 + \tfrac{(\tilde{b}tw)^2}{2}
    - \tfrac{(\tilde{b}tw)^4}{24}.
\]
By the alternating-series bound for cosine, this is
bounded in absolute value by
$(\tilde{b}tw)^6/720$, giving a Laplace contribution
of order $t^{6-5/Y}$.  Since $6 - 5/Y > 2/Y$ is equivalent
to $Y > 7/6$, the pure-drift residual is absorbed
into the remainder for $Y > 7/6$.
 
For $Y \leq 7/6$, the sixth-order drift term
at order $t^{6-5/Y}$ lies strictly between $t^{4-3/Y}$
(the quartic drift, already included) and $t^{2/Y}$ (the
binomial term).  Its coefficient is
\begin{align}
    a_{6,1} = \frac{\tilde{b}^6\,\sigma_Y^{-5/Y}}{720\,\pi\,Y}
    \,\Gamma\!\left(\frac{5}{Y}\right),
    \label{eq:a61_formula}
\end{align}
which follows the same pattern as $a_{2,1}$ and
$a_{4,1}$. More generally, the $k$-th even-order drift produces a 
term $t^{2k-(2k-1)/Y}$ with coefficient 
\begin{align}
    a_{2k,1} = \frac{(-1)^{k+1}\,\tilde{b}^{2k}\,
    \sigma_Y^{-(2k-1)/Y}}{(2k)!\,\pi\,Y}
    \,\Gamma\!\left(\frac{2k-1}{Y}\right),
    \qquad k \geq 1.
    \label{eq:d_drift_general}
\end{align}
The exponent in the term $a_{2k,1}\,t^{2k-(2k-1)/Y}$ is such that
$2k-(2k-1)/Y > 2/Y$ if and only if $Y > (2k+1)/(2k)$.
Thus, for any fixed $Y \in (1,2)$, only finitely many
drift terms fall below $t^{2/Y}$; their number is
$\mathcal{K}(Y) = \left\lfloor 1/(2(Y-1)) \right\rfloor$,
which grows without bound as $Y \to 1^+$.  Since $6-5/Y > 4-3/Y$ 
for all $Y > 1$ (as $2 > 2/Y$), each additional drift term is
strictly smaller than the preceding one; and since all
non-drift contributions are absorbed at order
$t^{6/Y}$ or $t^{4-2/Y}$, the only source of terms
at or before order $t^{2/Y}$ is this drift series.
 
For $Y > 7/6$, no additional drift terms are needed
and the remainder is
$o(t^{\max(2-1/Y,\,2/Y)})$ as claimed.  For
$Y \leq 7/6$, the expansion remains valid with
the remainder replaced by
$o(t^{\min(2/Y,\,2\mathcal{K}+2-(2\mathcal{K}+1)/Y)})$,
where $\mathcal{K} = \mathcal{K}(Y)$ and the drift coefficients
$a_{6,1}, a_{8,1}, \ldots, a_{2\mathcal{K},1}$ are included
as additional explicit terms via
\eqref{eq:d_drift_general}.
 
\medskip
\noindent\textbf{Step 4: Domain extension.}
The Laplace integrals computed in Step~3 are over the
truncated interval $[\Lambda, t^{-\rho}]$; we now extend
them to $(0,\infty)$ and bound the errors.
 
\emph{Lower extension $[0,\Lambda]$.}
On this bounded interval, all integrands from Step~3 are
smooth and bounded.  For example, the drift-squared
integrand satisfies
$\int_0^{\Lambda} t^2\,\tilde{b}^2\,e^{-\sigma_Y tw^Y}\,dw
\leq t^2\,\tilde{b}^2\,\Lambda = O(t^2)$.
Similarly, the binomial integrand satisfies
$\int_0^{\Lambda} t\,(1 - e^{-\sigma_Y tw^Y})\,w^{Y-3}\,dw
\leq t \int_0^{\Lambda} \sigma_Y t w^{2Y-3}\,dw
= O(t^2)$,
using $1 - e^{-x} \leq x$ for $x \geq 0$.  Since
$2 > \max(2-1/Y,\,2/Y)$, these are in the remainder.
 
\emph{Upper extension $[t^{-\rho},\infty)$.}
For the terms involving $e^{-\sigma_Y tw^Y}$ (the
drift-squared and quartic integrands), the extension
adds a super-exponentially small quantity by the
argument of Step~1, since
$e^{-\sigma_Y tw^Y} \leq e^{-\sigma_Y t^{-\varepsilon Y}}$
on $[t^{-\rho},\infty)$.
 
For the binomial integrand, which involves the factor
$1 - e^{-\sigma_Y tw^Y} \leq 1$, the extension adds
\[
    t\,\abs{\RE{\beta_1}}
    \int_{t^{-\rho}}^{\infty} w^{Y-3}\,dw
    = \frac{t\,\abs{\RE{\beta_1}}}{2-Y}\,
    (t^{-\rho})^{Y-2}
    = O\!\left(t^{2/Y + \varepsilon(2-Y)}\right).
\]
Since $\varepsilon(2-Y) > 0$, this exceeds $t^{2/Y}$
and is absorbed into the remainder.
 
\emph{Denominator replacement.}
Replacing $1/(w^2 + \tfrac{1}{4})$ with $1/w^2$ on
$I_{\mathrm{core}}$ introduces an error of
$1/(4w^2(w^2+\tfrac{1}{4})) \leq \posc\,w^{-4}$ in the
integrand.  For the drift-squared term, the additional
error is bounded by
$\posc\,t^2 \int_{\Lambda}^{\infty}
e^{-\sigma_Y tw^Y}\,w^{-2}\,dw = O(t^2)$,
and similarly for the other terms.  As before, $O(t^2)$
is in the remainder.
 
\medskip
\noindent\textbf{Conclusion.}
Combining Steps~1--4, for $Y>7/6$ the remainder $R_3(t)$ equals the
sum of the three extracted Laplace contributions, at
orders $t^{2-1/Y}$ (coefficient $a_{2,1}$),
$t^{4-3/Y}$ (coefficient $a_{4,1}$), and $t^{2/Y}$
(coefficient $a_{1,2}$), plus an error of $o(t^{\max(2-1/Y,\,2/Y)})$.
For $Y\le 7/6$, one must also include the additional drift terms
$a_{6,1}, a_{8,1}, \ldots, a_{2\mathcal{K},1}$ from \eqref{eq:d_drift_general}
before taking the remainder.
\end{proof}

\subsection{Numerical verification of the higher-order coefficients}
\label{subsec:formal_numerical}

We verify the higher-order coefficients $d_2$, $a_{2,1}$, and $a_{1,2}$ 
numerically.  For $d_2$, the integral formula 
\eqref{eq:formal_d2} was already compared with the closed-form 
expression $d_2^{FL}$ from \cite{fl_houdre_cgmy} in 
Figure~\ref{fig:diffs}, confirming agreement to the tolerance 
of the quadrature ($\sim 10^{-7}$).

For $a_{2,1}$ and $a_{1,2}$, we verify the formulas 
\eqref{eq:a21_formula} and \eqref{eq:a12_formula} by two 
independent methods, both based on deterministic quadrature.

\medskip\noindent\textbf{Method 1: Laplace integrals.}
We evaluate the Laplace integrals
\begin{align}
    L_{2,1}(t) &:= \frac{\tilde{b}^2\,t^2}{2\pi}\,
    \int_0^\infty 
    \frac{w^{2}\,e^{-\sigma_Y t w^Y}}
    {w^2 + \tfrac{1}{4}}\,dw,
    \label{eq:L21_laplace} \\
    L_{1,2}(t) &:= \frac{\Re{\beta_1}}{\pi}\,t\,
    \int_0^\infty 
    \frac{w^{Y-1}\bigl(1 - e^{-\sigma_Y t w^Y}\bigr)}
    {w^2 + \tfrac{1}{4}}\,dw,
    \label{eq:L12_laplace}
\end{align}
which retain the exact denominator $w^2 + \tfrac{1}{4}$ from the 
Lipton--Lewis representation rather than the large-$w$ 
approximations used in the derivation, and extract 
$L_{2,1}(t)/t^{2-1/Y}$ and $L_{1,2}(t)/t^{2/Y}$.  As 
$t \downarrow 0$, the factors $e^{-\sigma_Y tw^Y}$ and 
$(1 - e^{-\sigma_Y tw^Y})$ concentrate at 
$w \sim (\sigma_Y t)^{-1/Y} \to \infty$ where 
$(w^2 + \tfrac{1}{4})^{-1} \to w^{-2}$, so the ratios should 
converge to $a_{2,1}$ and $a_{1,2}$, respectively.

Table~\ref{tab:a21_laplace} confirms this for $a_{2,1}$, using 
parameter sets with $Y < 3/2$ where the drift-squared term is 
the leading third-order correction.

\begin{table}[ht]
\centering
\begin{tabular}{c c c c}
\hline
$t$ & $C{=}1,G{=}3,M{=}5$ & $C{=}1,G{=}3,M{=}5$ 
    & $C{=}1,G{=}3,M{=}5$ \\
    & $Y{=}1.2$ & $Y{=}1.3$ & $Y{=}1.4$ \\
\hline
$a_{2,1}$ (formula) & $0.008981$ & $0.015382$ & $0.027278$ \\
\hline
$10^{-2}$ & $0.95748$ & $0.94573$ & $0.93216$ \\
$10^{-3}$ & $0.99347$ & $0.99031$ & $0.98620$ \\
$10^{-4}$ & $0.99903$ & $0.99834$ & $0.99731$ \\
$10^{-5}$ & $0.99986$ & $0.99972$ & $0.99948$ \\
$10^{-6}$ & $0.99998$ & $0.99995$ & $0.99990$ \\
$10^{-7}$ & $1.00000$ & $0.99999$ & $0.99998$ \\
\hline
\end{tabular}
\caption{Ratio $L_{2,1}(t)/\bigl(a_{2,1}\,t^{2-1/Y}\bigr)$ for 
various parameter sets with $Y < 3/2$.  All ratios 
converge to $1$, confirming \eqref{eq:a21_formula}.}
\label{tab:a21_laplace}
\end{table}

Table~\ref{tab:a12_laplace} provides the analogous verification 
for $a_{1,2}$: for each parameter set, $L_{1,2}(t)/t^{2/Y}$ 
converges to the formula value with ratio approaching $1.000000$ 
as $t$ decreases.

\begin{table}[ht]
\centering
\begin{tabular}{c c c c c}
\hline
$t$ & $C{=}1,G{=}3,M{=}5$ & $C{=}1,G{=}3,M{=}5$ 
    & $C{=}1,G{=}3,M{=}5$ & $C{=}2,G{=}2,M{=}3$ \\
    & $Y{=}1.7$ & $Y{=}1.8$ & $Y{=}1.9$ & $Y{=}1.75$ \\
\hline
$a_{1,2}$ (formula) & $24.437$ & $29.730$ & $49.359$ & $36.297$ \\
\hline
$10^{-2}$ & $0.98782$ & $0.99140$ & $0.99396$ & $0.98291$ \\
$10^{-3}$ & $0.99794$ & $0.99864$ & $0.99905$ & $0.99710$ \\
$10^{-4}$ & $0.99967$ & $0.99980$ & $0.99987$ & $0.99956$ \\
$10^{-5}$ & $0.99995$ & $0.99997$ & $0.99998$ & $0.99994$ \\
$10^{-6}$ & $0.99999$ & $1.00000$ & $1.00000$ & $0.99999$ \\
$10^{-7}$ & $1.00000$ & $1.00000$ & $1.00000$ & $1.00000$ \\
\hline
\end{tabular}
\caption{Ratio $L_{1,2}(t)/\bigl(a_{1,2}\,t^{2/Y}\bigr)$ for 
various parameter sets and values of $t$.  All ratios 
converge to $1$ as $t \downarrow 0$, confirming the formula 
\eqref{eq:a12_formula}.}
\label{tab:a12_laplace}
\end{table}

\medskip\noindent\textbf{Method 2: Full Lipton--Lewis 
remainder.} We compute the remainder $R_3(t)$ directly from the 
combined-integrand formula \eqref{eq:R3_integral}, which avoids 
catastrophic cancellation by using $f(x) = 1 - e^x + x$ to 
subtract the first two orders exactly, and then verify the 
expansion \eqref{eq:formal_five_terms}, established 
in Theorem~\ref{thm:rigorous_high_order}, by subtracting the
higher-order terms.

For $Y > 3/2$, the leading third-order term is 
$a_{1,2}\,t^{2/Y}$, and we report 
$R_3(t)/(a_{1,2}\,t^{2/Y})$; this should converge to $1$.

For $Y < 3/2$, more care is needed.  The leading 
third-order term is $a_{2,1}\,t^{2-1/Y}$, but the \emph{next} 
candidate depends on $Y$: the formal exponent lattice 
places the cubic drift correction at $t^{3-2/Y}$, which 
lies below $t^{2/Y}$ when $Y < 4/3$ and above it when 
$Y > 4/3$.  Since $(i\tilde{b}\,w)^3$ is purely imaginary, 
the coefficient at $t^{3-2/Y}$ vanishes, and the actual 
next nonzero term after $a_{2,1}\,t^{2-1/Y}$ is 
$a_{4,1}\,t^{4-3/Y}$ for $Y < 5/4$ or $a_{1,2}\,t^{2/Y}$ for 
$Y > 5/4$.  We confirm this numerically by defining
\[
    R_4(t) := R_3(t) - a_{2,1}\,t^{2-1/Y}
\]
and checking that $R_4(t)/t^{3-2/Y} \to 0$ for $Y < 4/3$ 
(Table~\ref{tab:cubic_vanishes}), and that 
$R_4(t)/(a_{1,2}\,t^{2/Y}) \to 1$ for $5/4<Y<3/2$; for $1<Y<5/4$ 
one must instead check
\[
R_5(t)/(a_{1,2}\,t^{2/Y}) \to 1,
\]
where $R_5(t):=R_3(t)-a_{2,1}t^{2-1/Y}-a_{4,1}t^{4-3/Y}$ (see Table~\ref{tab:both_terms}). 

\begin{table}[ht]
\centering
\begin{tabular}{c c c c}
\hline
$t$ & $C{=}1,G{=}3,M{=}5$ & $C{=}1,G{=}3,M{=}5$ 
    & $C{=}1,G{=}3,M{=}5$ \\
    & $Y{=}1.15$ & $Y{=}1.2$ & $Y{=}1.3$ \\
\hline
$10^{-2}$ & $3.12$ & $5.12$ & $13.0$ \\
$10^{-3}$ & $1.85$ & $3.98$ & $16.4$ \\
$10^{-4}$ & $0.903$ & $2.56$ & $17.5$ \\
$10^{-5}$ & $0.385$ & $1.45$ & $16.8$ \\
$10^{-6}$ & $0.151$ & $0.767$ & $15.2$ \\
$10^{-7}$ & $0.056$ & $0.386$ & $13.3$ \\
\hline
\end{tabular}
\caption{$R_4(t)/t^{3-2/Y}$ for $Y < 4/3$.  The convergence 
toward zero confirms that the cubic drift coefficient at 
$t^{3-2/Y}$ vanishes (as expected, since 
$(i\tilde{b}\,w)^3$ is purely imaginary).  For $Y = 1.3$ 
the gap $2/Y - (3-2/Y) = 4/Y - 3 \approx 0.08$ is 
small, so the approach to zero is very slow.}
\label{tab:cubic_vanishes}
\end{table}

\begin{table}[ht]
\centering
\begin{tabular}{c c c c c c}
\hline
$t$ & $R_5/(a_{1,2}\,t^{2/Y})$
    & \multicolumn{2}{c}{$R_4(t)/(a_{1,2}\,t^{2/Y})$} 
    & \multicolumn{2}{c}{$R_3(t)/(a_{1,2}\,t^{2/Y})$} \\
    & $Y{=}1.2$ & $Y{=}1.3$ & $Y{=}1.4$ 
    & $Y{=}1.7$ & $Y{=}1.9$ \\
\hline
$10^{-2}$ & $0.242$ & $0.380$ & $0.493$ & $0.742$ & $0.886$ \\
$10^{-3}$ & $0.406$ & $0.572$ & $0.684$ & $0.865$ & $0.942$ \\
$10^{-4}$ & $0.562$ & $0.728$ & $0.821$ & $0.937$ & $0.973$ \\
$10^{-5}$ & $0.688$ & $0.834$ & $0.904$ & $0.973$ & $0.989$ \\
$10^{-6}$ & $0.781$ & $0.901$ & $0.950$ & $0.989$ & $0.996$ \\
$10^{-7}$ & $0.848$ & $0.941$ & $0.974$ & $0.996$ & $0.998$ \\
$10^{-8}$ &         &         & $0.988$ & $0.998$ & $0.999$ \\
\hline
\end{tabular}
\caption{Convergence to $a_{1,2}$.  For $Y = 1.2 < 5/4$
(first column): $R_5(t)/(a_{1,2}\,t^{2/Y})$ where
$R_5 = R_3 - a_{2,1}\,t^{2-1/Y} - a_{4,1}\,t^{4-3/Y}$;
the quartic drift correction is numerically negligible
for this parameter set ($a_{4,1} \approx -2 \times 10^{-5}$),
so the entries coincide with $R_4/(a_{1,2}\,t^{2/Y})$ to the
displayed precision.
For $5/4 < Y < 3/2$ (next two columns):
$R_4(t)/(a_{1,2}\,t^{2/Y})$ where
$R_4 = R_3 - a_{2,1}\,t^{2-1/Y}$.  For $Y > 3/2$ (last two
columns): $R_3(t)/(a_{1,2}\,t^{2/Y})$.  All columns converge
toward $1$.  For $Y < 4/3$, the convergence is slower because
the gap between $3-2/Y$ (whose coefficient vanishes) and
$2/Y$ is small.  Empty entries indicate quadrature
degradation.}
\label{tab:both_terms}
\end{table}

The underlying Laplace identities 
\eqref{eq:laplace_identity} and \eqref{eq:laplace_exp} 
were each verified independently to machine precision 
($\sim 10^{-14}$) by numerical quadrature with analytical tail 
handling, for all $Y$ in the range $(1.1, 1.95)$.

\subsection{Higher-order exponent structure}
\label{subsec:higher_order}

Beyond the coefficients $a_{2,1}$, $a_{4,1}$, and $a_{1,2}$ established
in Theorem~\ref{thm:rigorous_high_order}, the formal
Laplace method of Section~\ref{subsec:formal_d3} identifies additional
terms at higher exponents.  We briefly catalog those with exponent
at most $2$; their justification is deferred to future work.

\medskip\noindent\textbf{The $\kappa$ and cross-term at 
$t^{1+1/Y}$.}
Two mechanisms contribute at exponent $1 + 1/Y$.  First, the 
linear correction \eqref{eq:f_diff_linear} with 
$\delta(w) = \kappa$ produces a Laplace integral evaluating to 
$\kappa\,d_1$.  Second, the quadratic correction 
\eqref{eq:f_diff_quad} with the cross-term 
$2(i\tilde{b}\,w)(\beta_1\,w^{Y-1})$ gives an integral against 
$e^{-\sigma_Y tw^Y}$ at the same exponent.  Since $\beta_1$ is complex, 
$\Re(i\tilde{b}\,\beta_1) = -\tilde{b}\,\Im(\beta_1) \neq 0$ 
in general.  The combined coefficient is
\begin{align}
    d_{1+1/Y} = \frac{\kappa\,\sigma_Y^{1/Y}}{\pi}\,
    \Gamma\!\left(\frac{Y-1}{Y}\right)
    + \frac{\tilde{b}\,\Im(\beta_1)}{\pi Y}\,
    \sigma_Y^{-(Y-1)/Y}\,\Gamma\!\left(\frac{Y-1}{Y}\right).
    \label{eq:a_1p1Y}
\end{align}

\medskip\noindent\textbf{The second binomial and $\beta_1^2$ 
at $t^{3/Y}$.}
The linear correction with the $n = 2$ binomial term 
$\beta_2\,w^{Y-2}$ and the quadratic correction with 
$\beta_1^2\,w^{2Y-2}$ both contribute at exponent $3/Y$.  
The combined coefficient is
\begin{align}
    d_{3/Y} = -\frac{\Re(\beta_2)}{\pi Y}\,
    \sigma_Y^{(3-Y)/Y}\,\Gamma\!\left(1 - \frac{3}{Y}\right)
    - \frac{\Re(\beta_1^2)}{2\pi Y}\,
    \sigma_Y^{(3-2Y)/Y}\,\Gamma\!\left(\frac{2Y-3}{Y}\right),
    \label{eq:a_3Y}
\end{align}
where the first Laplace identity requires $1 - 3/Y \in (-1,0)$, 
i.e., $Y > 3/2$, and the second integral (against 
$e^{-\sigma_Y tw^Y}$) converges for 
$2Y - 3 > 0$, i.e., $Y > 3/2$.

\medskip\noindent\textbf{Effective ordering and bifurcation 
structure.}
For $Y$ near $3/2$, the exponents $2 - 1/Y$ and $2/Y$ coalesce 
at $Y = 3/2$ where $2 - 1/Y = 2/Y = 4/3$, and the exponents 
$1 + 1/Y$ and $3/Y$ are nearby ($5/3$ and $2$, 
respectively, at $Y = 3/2$). In practice, the ordering of terms 
by \emph{magnitude} may differ from the ordering by 
\emph{exponent}: for example, $a_{1,2}\,t^{2/Y}$ is numerically 
much larger than $a_{2,1}\,t^{2-1/Y}$ even when $2/Y > 2 - 1/Y$ 
(i.e., for $Y < 3/2$), because $\abs{a_{1,2}} \gg \abs{a_{2,1}}$.  The 
exponent lattice describes the asymptotic ordering as 
$t \to 0$; the practical ordering at a given maturity may be 
quite different.

The general bifurcation formula is
\begin{align}
    Y = \frac{n+j}{j}, \qquad n, j \geq 1,
    \label{eq:bifurcation}
\end{align}
obtained by equating the $n$-th binomial exponent $(n+1)/Y$ with 
the $j$-th drift-power exponent $j + (1-j)/Y$.  For 
$(n,j) = (1,2)$ this gives $Y = 3/2$; for $(n,j) = (1,3)$ 
it gives $Y = 4/3$.  However, since the cubic drift ($j = 3$) 
has a zero coefficient, the exponent $3 - 2/Y$ is absent from 
the expansion.  The effective next bifurcation after $Y = 3/2$ 
is at $Y = 5/4$, where the quartic drift $t^{4-3/Y}$ 
crosses the binomial $t^{2/Y}$.

The behavior at $Y = 3/2$ deserves particular comment, and it is 
best to follow along with 
Figure~\ref{fig:Y_order_selected}.  For $Y$ slightly below 
$3/2$, the fourth-order exponent is $2/Y$, which approaches 
$4/3$ as $Y \to 3/2^{-}$.  For $Y$ slightly above $3/2$, the 
fourth-order exponent is $2 - 1/Y$, which also approaches $4/3$ 
as $Y \to 3/2^{+}$.  In both cases, the limiting value $4/3$ is 
precisely the third-order exponent at $Y = 3/2$.  That is, from 
both sides the fourth-order exponent tries to coalesce with the 
third.  But at $Y = 3/2$ itself, the exponent $4/3$ is already 
occupied by the third-order term, so the fourth order is forced 
to jump to the next available candidate, $1 + 1/Y = 5/3$.  The 
fourth-order exponent is thus \emph{discontinuous} as a function 
of $Y$ at $Y = 3/2$: it is equal to $4/3 + \varepsilon$ for $Y$ 
near but different from $3/2$, and then jumps to $5/3$ at the 
critical value.

The formal method extends naturally to arbitrary order. At any 
order $N$, we truncate the binomial expansion 
\eqref{eq:M_term}--\eqref{eq:G_term} at $n = N$, expand 
the exponential $e^{h_0 + h}$ to the required order via the 
Taylor series, and expand the denominator 
$1/(v^2 + \varepsilon^2)$ to the required order.  Each 
resulting integral is either of the form 
$\int_0^\infty v^p\,e^{-\sigma_Y v^Y}\,dv$ or 
$\int_0^\infty v^p\,(1 - e^{-\sigma_Y v^Y})\,dv$, and in 
either case can be computed in 
closed form via gamma functions.

\section{Conclusion and Future Work}
\label{sec:conclusion}

We have derived the second-order ATM call-price expansion for 
the exponential CGMY model with $Y \in (1,2)$ directly from 
the characteristic function, via the Lipton--Lewis 
representation and a rescaling argument tied to the stable 
domain of attraction.  The resulting integral formula for $d_2$ 
(Theorem~\ref{thm:second_order_CGMY}) agrees numerically with 
the closed-form coefficient from \cite{figueroa_lopez_gong_houdre_2014}, 
confirming that the Fourier approach recovers the same 
second-order term without recourse to measure transformations 
or density expansions.

We then extended the combined-integrand method to extract 
closed-form coefficients beyond second order.  The 
drift-squared coefficient $a_{2,1}$ at $t^{2-1/Y}$ agrees 
with the expression $\tilde{\gamma}^2 p_Z(1,0)/2$ from 
\cite{fl_houdre_cgmy}, providing an independent derivation 
via a Laplace integral against $e^{-\sigma_Y tw^Y}$.  The 
binomial coefficient $a_{1,2}$ at $t^{2/Y}$ is, to our 
knowledge, a new closed-form result; its formula involves the 
real part of the first binomial correction $\beta_1$ and a 
gamma function evaluated at $1 - 2/Y$.  We also identified a 
quartic drift term $a_{4,1}$ at $t^{4-3/Y}$, proportional to 
$\tilde{b}^4$, which is intermediate between $a_{2,1}$ and 
$a_{1,2}$ for $Y < 5/4$.

An important structural finding is that odd powers of the 
drift $(i\tilde{b}\,w)^{2j+1}$ are purely imaginary and
contribute nothing to the expansion.  This eliminates the 
exponent $3 - 2/Y$ from the lattice and shifts the effective 
bifurcation from the naive prediction of $Y = 4/3$ to the 
corrected value $Y = 5/4$.  The resulting five-term expansion
is proved in Theorem~\ref{thm:rigorous_high_order},
where a dynamic cutoff partitions the integration domain into
inner, core, and tail regions to control the interplay between
the Taylor regime and the deep tail where the stable exponential
dominates.  The expansion is verified numerically to high 
precision across the full range of $Y$.

\subsection*{Future Work}

Several extensions of the present work appear natural.  The 
characteristic-function approach used here could be applied to 
other tempered stable families (for example, the generalized 
tempered stable models considered in 
\cite{figueroa_lopez_gong_houdre_2014}) to determine the 
minimal analytic conditions needed for a second-order term.  
Extending the analysis to the close-to-the-money regime, where 
log-moneyness $k_t \to 0$ alongside maturity, would require 
expanding the full Lipton--Lewis formula \eqref{eq:ll_form} 
rather than the ATM specialization \eqref{eq:ll_atm}; the 
rescaling methodology developed here provides a natural 
starting point.

On the formal side, the exponent lattice of 
Section~\ref{subsec:higher_order} identifies the terms at 
$t^{1+1/Y}$ and $t^{3/Y}$ as the next candidates beyond the 
five-term expansion, with closed-form coefficients involving 
$\kappa$, $\Im(\beta_1)$, $\Re(\beta_2)$, and 
$\Re(\beta_1^2)$.  Proving these and, in particular, 
establishing the correct remainder bounds at these 
orders would extend the expansion to six or seven terms 
in the low-$Y$ regime.  More ambitiously, one could ask whether 
the Laplace method employed here admits a systematic inductive 
argument that produces the $n$-th coefficient for arbitrary $n$.

Finally, since the integral representation of $d_2$ requires only 
the characteristic exponent and not a closed-form second-order 
constant, the methodology developed here extends naturally to 
models where such a constant is not available.

\appendix
\section{Proofs}

\begin{proof}[Proof of Proposition~\ref{prop:L_rep}]
    The identities \eqref{eq:cgmy_L} and \eqref{eq:ll_atm_trans} follow immediately 
    by making the substitution $v = u t^{1/Y}$ in equation \eqref{eq:ll_atm}. 
    For \eqref{eq:real_theta}, use the polar coordinate representation of complex 
    numbers to rewrite 
    \[
        \left( \Mt t^{1/Y} - iu \right)^{Y} \text{\; and \;}
        \left( \Gt t^{1/Y} + iu  \right)^{Y}.
    \]
\end{proof}

\begin{rem}
    \begin{enumerate}[label=(\roman*)]
    \item Note that $\theta\left( t, v \right) \rightarrow \theta_{0}\left( v \right)$ as 
    $t \rightarrow 0$ for every $v \geq 0$, and $\theta_{0}$ is a real-valued function as 
    \[
        \left( -i \right)^{Y} + \left( i \right)^{Y} = - 2 \abs{\cos{ \left( 
    \frac{Y \pi}{2} \right) } }.
    \]
    \item We will use certain substitutions frequently. To that end, note that
    \begin{align}
        \theta\left( t, t^{1/Y} v \right) = t \psi_0\left( v \right), 
        \label{eq:theta_sub} 
    \end{align}
    where
    \begin{align}
        \psi_0\left( v \right) = i v \tilde{b} + \kappa +
            C \Gamma\left( -Y \right)\left( \left( \Mt - iv \right)^{Y} + 
            \left( \Gt + iv \right)^{Y} \right),
    \end{align}
    and
    \begin{align}
        \theta_{0}\left( t ^{1/Y} v \right) = t \theta_{0}\left( v \right).
        \label{eq:theta0_sub}
    \end{align}
\item Using \eqref{eq:real_theta} and \eqref{eq:theta_sub} leads to 
    \begin{align}
        \Re{\left( \psi_0\left( v \right) \right)} &= \frac{r\left( t, t^{1/Y}v 
        \right)}{t} \notag \\
        &= \kappa + C \Gamma(-Y)   \left[ \left( \Mt^2 + v^2 \right)^{Y/2}   
        \cos \left( Y \arctan \left(
        -\frac{v}{\Mt} \right) \right) \right. \notag \\
        &\;\;\; + \left. \left( \Gt^2 + v^2 \right)^{Y/2}   \cos \left( Y \arctan \left(
        \frac{v}{\Gt } \right) \right) \right].
        \label{eq:real_psi}
    \end{align}
    Observe that $\Re{\left( \psi_0\left( v \right) \right)} \sim -2 C 
    \Gamma\left( -Y \right)
    \abs{\cos{\left( Y \pi/2 \right) } } v^{Y}$ as $v 
    \rightarrow \infty$ where the coefficient is negative, so that $\exp{\left( \Re{\left( 
    \psi_0\left( v \right) \right) } \right)}$ is bounded by $1$ for $v\geq 0$.
    \end{enumerate}
\end{rem}

\begin{proof}[Proof of Theorem \ref{thm:CGMY_first_order}]
    We proceed as in \cite{andersen_lipton} by considering the first order of $\calL$.  
    In order to prove the result, we break $\calL$ into two parts: one where the integration
    is restricted to the interval $[0,\varepsilon]$ and the other where the integration
    is restricted to $(\varepsilon,\infty)$, with $\varepsilon>0$ small. We denote these 
    two parts as $\LL{0}{\varepsilon}$ and $\LL{\varepsilon}{\infty}$, respectively. 
    In what follows, we will expand the integrand of $\LL{0}{\varepsilon}$ around 
    the origin, and we will apply Lebesgue's Dominated Convergence Theorem to 
    $\LL{\varepsilon}{\infty}$. We use $\posc$ to represent a 
    positive constant whose value can change from line to line in the remaining work.

    To first order, we have 
    \begin{align}
        \LL{0}{\varepsilon}\left( t \right) &= \frac{1}{\pi} \Re{\left( 
            \int_{0}^{\varepsilon} 
        \frac{ 1 - \exp{\left( \theta(t,v) \right)}}{v^2 + \frac{1}{4}t^{2/Y}} dv 
        \right) } \notag \\
        &= \frac{1}{\pi} \Re{\left( \int_{0}^{\varepsilon} 
        \frac{ -\theta(t,v) + D\left( t,v \right)}{v^2 + \frac{1}{4}t^{2/Y}} dv 
        \right) } 
        = - \frac{1}{\pi} \int_{0}^{\varepsilon} 
        \frac{ \Re{ \left( \theta(t,v) \right) } }{v^2 + \frac{1}{4}t^{2/Y}} dv
        + \frac{1}{\pi} \int_{0}^{\varepsilon} 
        \frac{ \Re{ \left( D(t,v) \right) } }{v^2 + \frac{1}{4}t^{2/Y}} dv,
        \label{eq:0_to_eps}
    \end{align}
    where $D\left( t,v \right) = O\left( \theta\left( t,v \right)^{2} \right)$ as
    $t,v \rightarrow 0$. 

    Here is where our proof differs from the one presented in \cite{andersen_lipton}: 
    the authors claim that \eqref{eq:0_to_eps} is $O( \varepsilon )$ after letting
    $t \rightarrow 0$ and ignore the remainder term 
    involving $D$. We were not able to verify this claim that
    \eqref{eq:0_to_eps} is $O\left( \varepsilon \right)$ as $t \rightarrow 0$, and 
    further we obtain $O\left( \varepsilon^{Y-1} \right)$ after letting $t 
    \rightarrow 0$.

    First, we show that the remainder term is $O\left( \varepsilon^{2Y-1} \right)$
    as $t \rightarrow 0$. We estimate, for some constant $\posc>0$ whose
    value might change from line to line, and make the 
    substitution $v = t^{1/Y} w$, to obtain
    \begin{align}
        \int_{0}^{\varepsilon} \abs{\frac{\Re{\left( D\left( t,v \right) \right)}}{
            v^{2} + \frac{1}{4} t^{2/Y} } }dv 
            &\leq \int_{0}^{\varepsilon} \frac{\posc \abs{\theta\left( t,v \right)}^{2}}{
            v^{2} + \frac{1}{4} t^{2/Y} } dv \notag \\
            &= \posc t^{-1/Y} \int_{0}^{\varepsilon t^{-1/Y}} \frac{\abs{\theta\left( 
                t,t^{1/Y} w \right)}^{2}}{w^{2} + \frac{1}{4} } dw \notag \\
                &= \posc t^{-1/Y} \int_{0}^{\varepsilon t^{-1/Y}} \frac{t^{2} \abs{
                \psi_0\left( w \right)}^{2}}{w^{2} + \frac{1}{4} } dw \notag \\
                &\leq \posc t^{2-1/Y} \int_{0}^{\varepsilon t^{-1/Y}} \frac{ 
                    \left( 1 \vee w^{2Y} \right)}{w^{2} + \frac{1}{4} } dw \notag \\
                    &\leq \posc t^{2-1/Y} \int_{0}^{1} \frac{1}{w^{2}+ \frac{1}{4}} dw +
                    \posc t^{2-1/Y} \int_{1}^{\varepsilon t^{-1/Y}} \frac{
                    w^{2Y}}{w^{2}} dw \notag \\
                    &\leq \posc t^{2-1/Y} + \posc t^{2-1/Y} \int_{1}^{\varepsilon t^{-1/Y}} 
                    w^{2Y-2} dw \notag \\
                    &\leq \posc t^{2-1/Y} + \frac{\posc}{2Y-1} \varepsilon^{2Y-1}, 
                    \label{eq:ll_fo_remainder}
    \end{align}
    which evaluates to $O\left( \varepsilon^{2Y-1} \right)$ (and hence $o\left( \varepsilon^{Y-1} \right)$) after letting $t \rightarrow 0$.

    Continuing the estimation of \eqref{eq:0_to_eps} and again 
    using the substitution $v = t^{1/Y} w$, 
    \begin{align}
        \abs{ \int_{0}^{\varepsilon} \frac{ \Re{ \left( \theta(t,v) \right) } }{v^2 
            + \frac{1}{4}t^{2/Y}} dv } &= \abs{ \int_{0}^{\varepsilon t^{-1/Y}} 
            \frac{ \Re{ \left( \theta(t,t^{1/Y} w) \right) } }{t^{2/Y} w^2 + 
            \frac{1}{4}t^{2/Y}} t^{1/Y} dw } \notag \\
            &= t^{-1/Y} \abs{ \int_{0}^{\varepsilon t^{-1/Y}} 
            \frac{ \Re{ \left( t \psi_0(w) \right) } }{w^2 + \frac{1}{4}} dw }
            \leq t^{1-1/Y}  \int_{0}^{\varepsilon t^{-1/Y}} 
            \frac{ \abs{\Re{ \left( \psi_0(w) \right) } } }{w^2 + \frac{1}{4}} dw,
                \label{eq:0_to_eps_est}
    \end{align}
    where $\psi_0$ is defined in \eqref{eq:psi_def}. Here, we can use 
    \eqref{eq:real_psi} to estimate the real part of $\psi_0$ as
    \begin{align}
        t^{1-1/Y}  \int_{0}^{\varepsilon t^{-1/Y}} 
            \frac{ \abs{\Re{ \left( \psi_0(w) \right) } } }{w^2 + \frac{1}{4}} dw
            &\leq t^{1-1/Y}  \int_{0}^{\varepsilon t^{-1/Y}} 
            \frac{ \abs{\kappa} + \posc \left( w^{Y} \vee 1 \right) }{w^2 + 
            \frac{1}{4}} dw.
        \label{eq:0_to_eps_final}
    \end{align}
    Splitting the integral, it is clear that the first part can be bounded by 
    \begin{align}
        t^{1-1/Y} \int_{0}^{\infty} \frac{\abs{\kappa}}{w^{2} + \frac{1}{4}}dw = 
        t^{1-1/Y} \abs{\kappa} \pi \rightarrow 0,
        \label{eq:0_to_eps_first}
    \end{align}
    as $t \rightarrow 0$. For the second part, split up the integral
    into two further parts: one integrating on the interval $[0,1]$ and one 
    integrating on the interval $\left(1, \varepsilon t^{-1/Y} \right)$. We assume 
    that $t$ is small enough that $\varepsilon t^{-1/Y} > 1$. On the interval 
    $[0,1]$, the integral can be estimated similarly to \eqref{eq:0_to_eps_first},
    so we only consider the interval $\left( 1, \varepsilon t^{-1/Y} \right)$. Then,
    \begin{align}
        t^{1-1/Y} \int_{1}^{\varepsilon t^{-1/Y}} \frac{\posc w^{Y}}{w^{2} + \frac{1}{4}} dw
        &\leq t^{1-1/Y} \int_{1}^{\varepsilon t^{-1/Y}} \frac{\posc w^{Y}}{w^{2}} dw \notag \\
        &= \frac{\posc}{Y-1} t^{1-1/Y} \left( \left( \varepsilon t^{-1/Y} \right)^{Y-1} 
         - 1 \right)
         = \frac{\posc}{Y-1} \varepsilon^{Y-1} - \frac{\posc}{Y-1} t^{1-1/Y}.
        \label{eq:0_to_eps_second}
    \end{align}
    Combining \eqref{eq:0_to_eps_first}, \eqref{eq:ll_fo_remainder}, and 
    \eqref{eq:0_to_eps_second} proves that \eqref{eq:0_to_eps} is 
    $O\left( \varepsilon^{Y-1} \right)$ as $t \rightarrow 0$. 

    For the integral $\LL{\varepsilon}{\infty}$, our proof is much simpler. Here, 
    $\abs{\exp{\theta\left( t, v \right)}}$ is bounded above by a constant, say
    by $\posc$. We estimate the integrand of $\LL{\varepsilon}{\infty}$ as
    \begin{align*}
        \abs{\frac{1 - \Re{\left( \exp{\left( \theta(t,v \right)} \right)} }{v^{2} + 
        \frac{1}{4} t^{2/Y}}} &\leq \frac{1 + \posc}{v^2} \in L^{1}\left(
        \varepsilon,\infty \right), 
    \end{align*}
    where $\posc$ possibly depends on $\varepsilon$. 
    Thus, we apply Lebesgue's Dominated Convergence Theorem to obtain that 
    \begin{align}
        \lim_{t \rightarrow 0} \LL{\varepsilon}{\infty}\left( t \right) =
        \frac{1}{\pi}\int_{\varepsilon}^{\infty} \frac{1 - \exp{\left( \tht{v} \right)}}{v^{2}} dv.
        \label{eq:eps_inf_dct}
    \end{align}

    Finally, we have
    \begin{align}
        \limsup_{t \rightarrow 0} \biggl\lvert \calL\left( t \right) &- 
        \frac{1}{\pi} \int_{0}^{\infty} 
        \frac{1 - \exp{\left( \tht{v} \right)}}{v^{2}} dv \biggr\rvert \notag \\ 
        &= \limsup_{t \rightarrow 0} \abs{\LL{0}{\varepsilon}\left( t \right) 
        + \LL{\varepsilon}{\infty}\left( t \right) - \frac{1}{\pi} \int_{0}^{\infty} 
        \frac{1 - \exp{\left( \tht{v} \right)}}{v^{2}} dv } \notag \\
        &\leq \limsup_{t \rightarrow 0} \abs{\LL{0}{\varepsilon}\left( t \right)} +
        \frac{1}{\pi}\abs{ \int_{0}^{\varepsilon} \frac{1 - \exp{\left( \tht{v} \right)}}{v^{2}} dv } 
        \leq \posc \varepsilon^{Y-1} + \frac{1}{\pi}\abs{ \int_{0}^{\varepsilon} \frac{1 - 
            \exp{\left( \tht{v} \right)}}{v^{2}} dv }, \label{eq:LL_est} 
    \end{align}
    and \eqref{eq:LL_est} converges to $0$ as $\varepsilon \rightarrow 0$. 
    \end{proof}

    We present two technical results that we will need in the proof of the 
    second-order expansion.
    
\begin{lem}
 \label{lem:exp_diff}
    Let $f,g:\left[0, \infty \right) \rightarrow \C$ be functions such that
    $\Re{f(w)}, \Re{g(w)} < 0$ for $w$ large enough. Then for any $t>0$ and 
    $w$ large enough,
    $$
        \abs{e^{t f(w)} - e^{t g(w)}} \leq t \abs{f(w) - g(w)}.
    $$    
\end{lem}

\begin{proof}
    Define the function $h(s) = e^{t\left( s f(w) + (1-s) g(w)\right)}$ where 
    $s\in [0,1]$. Then, we have
    \[
    e^{t f(w)} - e^{t g(w)} = h(1) - h(0) = \int_0^1 h'(s) ds.
    \]
    Note that 
    \[h'(s) = t \left( f(w) - g(w) \right) 
    e^{t\left( s f(w) + (1-s) g(w)\right)}.\] 
    Finally, since $\Re{\left(s f(w) + (1-s) g(w)\right)} < 0$ for $w$
    large enough, we obtain
    \begin{align*}
        \abs{e^{t f(w)} - e^{t g(w)}} &= \abs{\int_0^1 h'(s) ds} \\
        &\leq \int_0^1 \abs{h'(s)} ds \\
        &= t \abs{\left( f(w) - g(w) \right)} \int_0^1 
        \abs{e^{t\left( s f(w) + (1-s) g(w)\right)}} ds \\
        &= t \abs{f(w) - g(w)} \int_0^1 
        e^{t \Re{\left(s f(w) + (1-s) g(w)\right)} } ds \\
        &\leq t \abs{f(w) - g(w)} \int_0^1 ds
        = t \abs{f(w) - g(w)},
    \end{align*}
    where we used $\abs{e^{w}} = e^{\Re{w}} \leq 1$ for $\Re{w} \leq 0$.
\end{proof}

\begin{prop}
 \label{prop:A11_est}
    There exists $w_0 > 0$ and $\eta(w_0) > 0$ such that for all $w > w_0$
    \begin{align}
        \abs{\RE{\tht{w} - \psi_0 \left( w \right)}} \leq \eta w^{Y-1} \vee \abs{\kappa}. 
    \label{eq:tht0_psi_close}
    \end{align}
\end{prop}

\begin{proof}
    The difference $\tht{w} - \psi_0(w)$ consists of the drift 
    $-i\tilde{b}\,w$, the constant $-\kappa$, and two binomial corrections 
    of the form 
    \begin{align}
        g(u) = \left( B \pm i u \right)^Y - \left( \pm i u \right)^Y
    \end{align}
    where $B$ is a real value ($\Mt$ or $\Gt$).  The drift term 
    $-i\tilde{b}\,w$ is purely imaginary and vanishes under $\Re$, so only 
    the binomial corrections $g$ and the constant $\kappa$ contribute to 
    the real part.  Since $\kappa$ enters additively and 
    $\eta w^{Y-1} \rightarrow \infty$ as $w\rightarrow \infty$, it suffices 
    to take the maximum on the right-hand side of \eqref{eq:tht0_psi_close}. 
    Thus, we just need to show that 
    \begin{align}
        \abs{g(u)} \leq \eta \abs{u}^{Y-1}.
    \end{align}
    To that end, consider for large $u$ and using an analytic expansion for 
    $(1+z)^Y$ when $z \in \C$ is close to $0$, 
    \begin{align}
        g\left( u \right) &=
        \left(\pm iu\right)^Y \left[ \left(1 \pm \frac{B}{iu} \right)^Y 
         - 1 \right] \notag \\
         &= \left(\pm iu\right)^Y \left[ 1 \pm Y \frac{B}{iu} + o\left(
         u^{-1} \right)
         - 1 \right] 
         =   Y B \left( \pm i u\right)^{Y-1} + o\left( u^{Y-1} \right). 
         \label{eq:g_asymptotic}
    \end{align}
    Taking the modulus gives
    $$
    \abs{g(u)} \leq \eta \abs{u}^{Y-1} 
    $$
    for $u$ large enough. 
\end{proof}

\begin{proof}[Proof of Theorem~\ref{thm:second_order_CGMY}]
    We obtain \eqref{eq:cgmy_secondorder} by showing that the function
    \begin{align}
        R\left( t \right) := \frac{c\left( t, 0 \right)}{t^{1/Y}} - \calL\left( 0 \right),
        \label{eq:cgmy_second_remainder}
    \end{align}
    is of order $t^{1-1/Y}$. In terms of $\calL$, write
    \begin{align}
        R\left( t \right) &= \frac{1}{\pi} \Re{\left( \int_{0}^{\infty} \frac{ 1 - 
        \exp{\left( \theta\left( t, v \right) \right) }}{ v^{2} + \frac{1}{4} t^{2/Y} 
    } dv \right) } - \frac{1}{\pi} \int_{0}^{\infty} \frac{1 - 
        \exp{ \left( \tht{ v } \right) } }{ v^2 } dv \notag \\
    &= \frac{1}{\pi} \Re{\left( \int_{0}^{\infty} \frac{ 1 - 
        \exp{\left( \theta\left( t, v \right) \right) }}{ v^{2} + \frac{1}{4} t^{2/Y} 
    } dv  - \int_{0}^{\infty} \frac{1 - 
        \exp{ \left( \tht{ v } \right) } }{ v^2  } dv \right) }. \label{eq:remain_comb}
    \end{align}
    In \eqref{eq:remain_comb}, temporarily ignore the $1/\pi$ and real part and just consider
    the two integrals. 
    
    Fix $\varepsilon> 0 $ and consider the two regions $\left\{  0 < 
    t/v^{Y} < \varepsilon \right\}$ and $\left\{ t/v^{Y} \geq \varepsilon
    \right\}$. Splitting \eqref{eq:remain_comb}, we obtain
    \begin{align}
        R\left( t \right) &= \left( \int_{\left( t/\varepsilon \right)^{1/Y}}^{\infty} 
        + \int_{0}^{\left( t/\varepsilon \right)^{1/Y}}\right) \left( 
        \frac{1 - \exp{\left( \theta\left( t, v \right) \right)}}{v^{2} + \frac{1}{4} 
        t^{2/Y}} - \frac{1 - \exp{\left( \tht{ v } \right)}}{v^2} \right) dv \notag \\
        &= A_{1}\left( t, \varepsilon \right) + A_{2}\left( t, \varepsilon \right) .
        \label{eq:remain_parts}
    \end{align}
 
    First, we evaluate $A_{2}$ by combining fractions and making the substitution
    $v = t^{1/Y} w$, 
    \begin{align}
    A_{2}&\left( t, \varepsilon \right) = \int_{0}^{\left( t/\varepsilon \right)^{1/Y}}
    \left( \frac{ v^{2} \left( 1 - \exp{\left( \theta\left( t, v \right) \right)} \right) 
    - \left(  v^{2} + \frac{1}{4} t^{2/Y} \right) \left( 1 - \exp{\left( \tht{v}
    \right) } \right) }{v^{2} \left( v^{2} + \frac{1}{4} t^{2/Y} \right) } \right) dv 
    \notag \\
    &= t^{-1/Y} \int_{0}^{\varepsilon^{-1/Y}}
    \frac{ w^{2} \left( 1 - \exp{\left( \theta\left( t, t^{1/Y} w \right) 
    \right)} \right) - \left(  w^{2} + \frac{1}{4} \right) \left( 1 - \exp{\left( 
        \tht{t^{1/Y} w}
    \right) } \right) }{w^{2} \left( w^{2} + \frac{1}{4} \right) } dw \notag \\
    &= t^{-1/Y} \int_{0}^{\varepsilon^{-1/Y}}
    \left( \frac{ w^{2} \left( 1 - \exp{\left( t \psi_0\left( w \right) 
    \right)} \right) - \left(  w^{2} + \frac{1}{4} \right) \left( 1 - \exp{\left( t
        \tht{w}
    \right) } \right) }{w^{2} \left( w^{2} + \frac{1}{4} \right) } \right)  dw.
    \label{eq:A2_precise} 
    \end{align}
    Notice that for $w$ close to $0$ and $t$ close to $0$, we can formally expand the 
    numerator to first order and arrive at
    \begin{align}
        t^{1-1/Y} \int_{0}^{\varepsilon^{-1/Y}}
    \left( \frac{ \left(  w^{2} + \frac{1}{4} \right) \tht{w}- w^{2} \psi_0\left( w 
    \right) }{w^{2} \left( w^{2} + \frac{1}{4} \right) } \right)  dw. \label{eq:A2}
    \end{align}
    The integral in \eqref{eq:A2}, considered on its own, is in fact well-defined
    after taking real parts, since $\Re{\psi_0\left( w \right)} \to \kappa$ as 
    $w \to 0$ and $\tht{w}/w^2 = -\sigma_Y w^{Y-2}$ is integrable.
    Using Lebesgue's Dominated Convergence Theorem, we can show \eqref{eq:A2} holds 
    precisely, i.e.
    \begin{align}
        \lim_{t \rightarrow 0} \frac{\Re{ \left( A_{2}\left( t, \varepsilon \right) 
        \right) } }{t^{1-1/Y}} = \int_{0}^{\varepsilon^{-1/Y}}
    \left( \frac{ \left(  w^{2} + \frac{1}{4} \right) \tht{w}- w^{2} 
    \Re{ \left( \psi_0\left( w 
    \right) \right) } }{w^{2} \left( w^{2} + \frac{1}{4} \right) } \right)  dw. 
    \label{eq:A2_conv}
    \end{align}
 
    We start with the representation \eqref{eq:A2_precise} and compute
    \begin{align}
        \frac{ A_{2}\left( t, \varepsilon \right)}{t^{1-1/Y}} &= 
        t^{-1} \int_{0}^{\varepsilon^{-1/Y}}
        \left( \frac{ w^{2} \left( 1 - \exp{\left( t \psi_0\left( w \right) 
        \right)} \right) - \left(  w^{2} + \frac{1}{4} \right) \left( 1 - \exp{\left( t
            \tht{w}
        \right) } \right) }{w^{2} \left( w^{2} + \frac{1}{4} \right) } \right)  dw \notag \\
        &= t^{-1} \int_{0}^{\varepsilon^{-1/Y}}
        \left( \frac{ 1 - \exp{\left( t \psi_0\left( w \right) 
        \right)} }{\left( w^{2} + \frac{1}{4} \right) } \right)  dw - 
        t^{-1} \int_{0}^{\varepsilon^{-1/Y}}
        \left( \frac{ 1 - \exp{\left( t \tht{w}
        \right) }  }{w^{2} } \right)  dw  \notag \\
        &= A_{21}\left( t, \varepsilon \right) - A_{22}\left( t, \varepsilon \right).
        \label{eq:A2_split}
    \end{align}
    First, we estimate the integrand of $A_{22}$ as 
    \begin{align}
        \abs{ \frac{1}{t} \Re{ \left( \frac{ 1 - \exp{\left( t \tht{w} \right) }  
        }{w^{2} } \right) } } 
        &\leq \frac{1}{t} \frac{ \abs{ 1 - \exp{\left( t \tht{w} \right) } } 
        }{w^{2} } 
        \leq \frac{\abs{ \tht{w} } }{w^2}
        \leq \frac{ \posc w^{Y}}{w^2}, \label{eq:A22_bound}
    \end{align}
    where $\posc >0$ and \eqref{eq:A22_bound} is in $L^{1}\left[ 0, \varepsilon^{-1/Y} 
    \right]$. Noting that, for every $w \geq 0$, 
    \[
        \lim_{t \rightarrow 0} \frac{ 1 - \exp{\left( t \tht{w} \right) }  }{t} =
        - \tht{w},
    \]
    Lebesgue's Dominated Convergence Theorem gives
    \[
        \lim_{t \rightarrow 0} A_{22}\left( t, \varepsilon \right) = 
        - \int_{0}^{\varepsilon^{-1/Y}} \frac{\tht{w}}{w^2} dw.
    \]
    We can apply a similar argument to $A_{21}$ with the one exception being that 
    our bounding function is now
    \[
        \frac{2 \abs{\psi_0\left( w \right)}}{w^2 + \frac{1}{4}},
    \]
    where $\abs{\psi_0}$ and $1/\left( w^{2} + 1/4 \right)$ are bounded 
    on $\left[ 0, \varepsilon^{-1/Y} \right]$. Again, using Lebesgue's Dominated
    Convergence Theorem and recombining the results gives
    \[
        \lim_{t \rightarrow 0} \frac{ A_{2}\left( t, \varepsilon \right)}{ t^{1-1/Y}} = 
        \int_{0}^{\varepsilon^{-1/Y}}
        \left( \frac{ \left(  w^{2} + \frac{1}{4} \right) \tht{w}- w^{2} 
        \Re{\left( \psi_0\left( w \right) \right)} }{w^{2} \left( w^{2} + \frac{1}{4} \right) } \right)  dw.
    \]
 
    We now turn our attention to $A_1$.  Using the identity
    \begin{align}
        \frac{1}{v^{2} + \frac{1}{4} t^{2/Y}} - \frac{1}{v^{2}} 
        = \frac{-\frac{1}{4} t^{2/Y}}{v^{2}\!\left( v^{2} 
        + \frac{1}{4} t^{2/Y} \right)},
        \label{eq:pf_identity}
    \end{align}
    we decompose the integrand of $A_1$ as
    \begin{align}
        &\frac{1 - \exp{\left( \theta\left( t, v \right) \right)}}{v^{2} 
        + \frac{1}{4} t^{2/Y}} - \frac{1 - \exp{\left( \tht{v} \right)}}{v^{2}} 
        \notag \\
        &\qquad = \frac{\exp{\left( \tht{v} \right)} 
        - \exp{\left( \theta\left( t, v \right) \right)}}{v^{2} + \frac{1}{4} t^{2/Y}}
        \;-\; \frac{\frac{1}{4} t^{2/Y} \left( 1 - \exp{\left( \tht{v} \right)} 
        \right)}{v^{2}\!\left( v^{2} + \frac{1}{4} t^{2/Y} \right)}.
        \label{eq:A1_identity}
    \end{align}
    This is an exact algebraic identity, so
    \begin{align}
        A_{1}\left( t,\varepsilon \right) 
        &= \int_{\left( t/\varepsilon \right)^{1/Y}}^{\infty} 
           \frac{\exp{\left( \tht{v} \right)} - \exp{\left( \theta\left( t, v \right) 
           \right)}}{v^{2} + \frac{1}{4} t^{2/Y}}\,dv \notag \\
        &\quad - \frac{1}{4} t^{2/Y} \int_{\left( t/\varepsilon \right)^{1/Y}}^{\infty} 
           \frac{1 - \exp{\left( \tht{v} \right)}}
                {v^{2}\!\left( v^{2} + \frac{1}{4} t^{2/Y} \right)}\,dv \notag \\
        &= A_{11}\left( t, \varepsilon \right) + A_{12}\left( t, \varepsilon \right),
        \label{eq:first_part_parts}
    \end{align}
    with no approximation error.
 
    We begin with $A_{12}$.  The substitution $v = t^{1/Y} w$ gives
    \begin{align}
        \frac{A_{12}\left( t, \varepsilon \right)}{t^{1-1/Y}} 
        = -\frac{1}{4t} \int_{\varepsilon^{-1/Y}}^{\infty} 
        \frac{1 - \exp{\left( t \tht{w} \right)}}{w^{2}\!\left( w^{2} 
        + \frac{1}{4} \right)}\,dw.
        \label{eq:A12_rescaled}
    \end{align}
    Using the elementary inequality $\abs{e^z - 1} \leq \abs{z}$ for 
    $\Re{z} \leq 0$, with $z = t\tht{w}$ 
    (noting $\tht{w} \leq 0$ for all $w \geq 0$),
    \begin{align}
        \abs{\frac{1 - \exp{\left( t \tht{w} \right)}}{t \cdot w^{2}\!\left( w^{2} 
        + \frac{1}{4} \right)}} 
        \leq \frac{\abs{\tht{w}}}{w^{2}\!\left( w^{2} + \frac{1}{4} \right)}
        \leq \frac{\posc\, w^{Y}}{w^{4}} = \posc\, w^{Y-4},
        \label{eq:A12_bound}
    \end{align}
    which is in $L^{1}\!\left[ \varepsilon^{-1/Y}, \infty \right)$ since 
    $Y - 4 \in (-3,-2)$.  Since 
    $\left( 1 - e^{t\tht{w}} \right)/t \rightarrow -\tht{w}$ 
    pointwise as $t \rightarrow 0$, Lebesgue's Dominated 
    Convergence Theorem gives
    \begin{align}
        \lim_{t \rightarrow 0} \frac{\Re{\left( A_{12}\left( t, \varepsilon \right)
        \right)}}{t^{1-1/Y}} = \frac{1}{4} \int_{\varepsilon^{-1/Y}}^{\infty} 
        \frac{\tht{w}}{w^{2}\!\left( w^{2} + \frac{1}{4} \right)}\,dw.
        \label{eq:A12_coefficient}
    \end{align}
 
    We now consider $A_{11}$.  The substitution $v = t^{1/Y} w$ gives
    \begin{align}
        \frac{A_{11}\left( t, \varepsilon \right)}{t^{1-1/Y}} 
        = t^{-1} \int_{\varepsilon^{-1/Y}}^{\infty} 
        \frac{\exp{\left( t \tht{w} \right)} - \exp{\left( t \psi_0\left( w \right) 
        \right)}}{w^{2} + \frac{1}{4}}\,dw.
        \label{eq:A11_rescaled}
    \end{align}
    For $\varepsilon$ small enough that $\varepsilon^{-1/Y} > w_{0}$ (with 
    $w_{0}$ as in Proposition~\ref{prop:A11_est}), the hypotheses of 
    Lemma~\ref{lem:exp_diff} are satisfied for all $w \geq \varepsilon^{-1/Y}$, so
    \[
        \abs{e^{t \tht{w}} - e^{t \psi_0\left( w \right)}} 
        \leq t\,\abs{\tht{w} - \psi_0\left( w \right)}.
    \]
    Since $A_{11}$ appears inside $\RE{\cdot}$ in \eqref{eq:remain_comb}, 
    it suffices to bound the real part of the integrand.  Write
    \[
        \RE{e^{t\tht{w}} - e^{t\psi_0(w)}}
        = \bigl(e^{t\tht{w}} - e^{t\RE{\psi_0(w)}}\bigr)
        + e^{t\RE{\psi_0(w)}}\bigl(1 - \cos\bigl(t\,\Im{\psi_0(w)}\bigr)\bigr).
    \]
    For the first term, both $\tht{w}$ and $\RE{\psi_0(w)}$ are negative for 
    $w \geq \varepsilon^{-1/Y}$, so 
    \[
    \abs{e^{t\tht{w}} - e^{t\RE{\psi_0(w)}}} \leq t\abs{\RE{\tht{w} - \psi_0(w)}}
    \leq t\bigl(\eta\,w^{Y-1} \vee \abs{\kappa}\bigr)
    \]
    by Proposition~\ref{prop:A11_est}.  For the second, 
    $1 - \cos(x) \leq x^{2}/2$ and $\Im{\psi_0(w)} = O(w)$ give 
    $e^{t\RE{\psi_0}}\bigl(1 - \cos(t\,\Im{\psi_0})\bigr) 
    \leq \posc\,t^{2}\,w^{2}\,e^{-\posc\,t w^{Y}}$,
    and since $\sup_{t>0}\, t\,e^{-\posc\,t w^{Y}} \leq \posc\,w^{-Y}$, we obtain
    \begin{align}
        \frac{\abs{\RE{e^{t \tht{w}} - e^{t \psi_0\left( w \right)}}}}
             {t\!\left( w^{2} + \frac{1}{4} \right)}
        \leq \frac{\eta\, w^{Y-1} + \posc\, w^{2-Y}}
             {w^{2} + \frac{1}{4}},
        \label{eq:A11_bound}
    \end{align}
    which is in $L^{1}\!\left[ \varepsilon^{-1/Y}, \infty \right)$ since 
    $Y - 3 \in (-2,-1)$ and $-Y < -1$.  Since 
    $\RE{\left( e^{t\tht{w}} - e^{t\psi_0(w)} \right)}/t 
    \rightarrow \tht{w} - \RE{\psi_0\left( w \right)}$ 
    pointwise, Lebesgue's Dominated Convergence Theorem gives
    \begin{align}
        \lim_{t \rightarrow 0} \frac{\RE{\left( A_{11}\left( t, \varepsilon \right) 
    \right)}}{t^{1-1/Y}} &= \int_{\varepsilon^{-1/Y}}^{\infty} \frac{ \tht{w} 
        - \RE{\left( \psi_0\left( w \right) \right)} }{ w^{2} + \frac{1}{4} }\,dw.
        \label{eq:A11_formal}
    \end{align}
    The integral in \eqref{eq:A11_formal} is well-defined: although 
    $\tht{w}$ and $\Re{\left( \psi_0\left( w \right) \right)}$ each 
    grow like $w^{Y}$, their difference is $O\!\left( w^{Y-1} \right)$ 
    by Proposition~\ref{prop:A11_est}, so the integrand is 
    $O\!\left( w^{Y-3} \right)$.
 
    We now collect the terms.  Combining 
    \eqref{eq:A11_formal} and \eqref{eq:A12_coefficient},
    \begin{align}
        \lim_{t \rightarrow 0} \frac{\Re{\left( A_{1}\left( t, \varepsilon \right) 
        \right)}}{t^{1-1/Y}} 
        &= \int_{\varepsilon^{-1/Y}}^{\infty} \frac{ \tht{w} 
        - \Re{\left( \psi_0\left( w \right) \right)} }{ w^{2} + \frac{1}{4} }\,dw
        \;+\; \frac{1}{4} \int_{\varepsilon^{-1/Y}}^{\infty} 
        \frac{\tht{w}}{w^{2}\!\left( w^{2} + \frac{1}{4} \right)}\,dw \notag \\
        &= \int_{\varepsilon^{-1/Y}}^{\infty} 
        \frac{\left( w^{2} + \frac{1}{4} \right) \tht{w} - w^{2} 
        \Re{\left( \psi_0\left( w \right) \right)}}{w^{2}\!\left( w^{2} 
        + \frac{1}{4} \right)}\,dw.
        \label{eq:d2_a1}
    \end{align}
    Together with the $A_{2}$ limit \eqref{eq:A2_conv}, we obtain
    \begin{align}
        \lim_{t\to 0} \frac{R(t)}{t^{1-1/Y}} 
        &= \frac{1}{\pi} \left[
        \int_{\varepsilon^{-1/Y}}^{\infty} 
        \frac{\left( w^{2} + \frac{1}{4} \right) \tht{w} - w^{2} 
        \Re{\left( \psi_0\left( w \right) \right)}}{w^{2}\!\left( w^{2} 
        + \frac{1}{4} \right)}\,dw \right. \notag \\
        &\qquad\quad + \left.
        \int_{0}^{\varepsilon^{-1/Y}} 
        \frac{\left( w^{2} + \frac{1}{4} \right) \tht{w} - w^{2} 
        \Re{\left( \psi_0\left( w \right) \right)}}{w^{2}\!\left( w^{2} 
        + \frac{1}{4} \right)}\,dw \right] \notag \\
        &= \frac{1}{\pi} \int_{0}^{\infty} 
        \frac{\left( w^{2} + \frac{1}{4} \right) \tht{w} - w^{2} 
        \Re{\left( \psi_0\left( w \right) \right)}}{w^{2}\!\left( w^{2} 
        + \frac{1}{4} \right)}\,dw.
        \label{eq:d2_a2}
    \end{align}
    The integrand is $O\!\left( w^{Y-2} \right)$ as $w \rightarrow 0$ 
    (since $\tht{w} = -\sigma_{Y} w^{Y}$ and 
    $\Re{\left( \psi_0\left( w \right) \right)} \rightarrow \kappa$ as 
    $w \rightarrow 0$) and $O\!\left( w^{Y-3} \right)$ as 
    $w \rightarrow \infty$ (by Proposition~\ref{prop:A11_est}), so 
    the integral converges and equals \eqref{eq:d2_formal}.
\end{proof}

\bibliography{levywork_updated}

\end{document}